\documentclass[12pt]{article}
\usepackage[left=1in,right=1in]{geometry}
\usepackage{amsmath,amsfonts,amssymb,mathrsfs,array,stmaryrd,indentfirst,amsthm,comment,color,mathtools,bbm,commath}
\SetSymbolFont{stmry}{bold}{U}{stmry}{m}{n}
\usepackage{authblk}
\usepackage{graphicx,hyperref,enumitem,tabularx}
\definecolor{linknavy}{RGB}{20,40,120}
\definecolor{citegreen}{RGB}{20,110,60}
\hypersetup{
  colorlinks=true,
  linkcolor=linknavy,
  citecolor=citegreen,
  urlcolor=linknavy,
}
\usepackage{verbatim}
\usepackage{upgreek}
\usepackage{float}

\usepackage{natbib}

\newtheorem{theorem}{Theorem}[section]
\newtheorem{lemma}[theorem]{Lemma}
\newtheorem{proposition}[theorem]{Proposition}
\newtheorem{corollary}[theorem]{Corollary}
\newtheorem{remark}[theorem]{Remark}

\theoremstyle{definition}
\newtheorem{definition}[theorem]{Definition}

\numberwithin{equation}{section}

\newcommand{\cA}{\mathcal{A}}

\newcommand{\cH}{\mathcal{H}}

\newcommand{\cM}{\mathcal{M}}
\newcommand{\cN}{\mathcal{N}}
\newcommand{\cO}{\mathcal{O}}

\newcommand{\cS}{\mathcal{S}}

\def\cN{\mathcal{N}}

\newcommand{\vt}{{\vartheta}}
\newcommand{\Om}{{\Omega}}

\def \e{\varepsilon}

\def \1{\mathbf 1}

\def \vt {{\bf \tilde v}}
\def \bfSigmat {\boldsymbol {\tilde{\Sigma}}}

\def \bA {{\bf A}}
\def \bB {{\bf B}}
\def \bC {{\bf C}}

\def \bG {{\bf G}}
\def \bI {{\bf I}}
\def \bL {{\bf L}}
\def \bM {{\bf M}}
\def \bN {{\bf N}}
\def \bO {{\bf O}}
\def \bP {{\boldsymbol P}}
\def \bQ {{\bf Q}}
\def \bR {{\bf R}}

\def \bU {{\bf U}}
\def \bV {{\bf V}}
\def \bW {{\boldsymbol W}}
\def \bX {{\boldsymbol X}}
\def \bY {\boldsymbol Y}
\def \bZ {{\boldsymbol Z}}

\def \be {{\bf e}}
\def \bm {{\bf m}}
\def \bp {{\bf p}}

\def \bu {{\bf u}}
\def \bv {{\bf v}}

\def \bx {{\bf x}}
\def \by {{\bf y}}
\def \bz {{\bf z}}

\def \bfH {{\boldsymbol H}}
\newcommand{\bfP}{\mathbf{P}}

\def \bfgamma {{\boldsymbol \gamma}}
\def \bfalpha {{\boldsymbol \alpha}}

\def \bfSigma {\boldsymbol \Sigma}
\def \bfeta {\boldsymbol \upeta}
\def \bfsigma {\boldsymbol \upsigma}
\def \bfxi {\boldsymbol \xi}

\def \bfLambda {{\boldsymbol \Lambda}}

\def \sF {\mathscr{F}}

\def\dE{\mathbb{E}}

\def\dP{\mathbb{P}}
\def\dQ{\mathbb{Q}}
\def\dR{\mathbb{R}}

\newcommand{\E}[1]{\mathbb{E}\!\sbr{#1}}
\newcommand{\Epi}[1]{\mathbb{E}^{\pi}\!\sbr{#1}}
\newcommand{\var}[1]{\mathrm{Var}\!\del{#1}}
\newcommand{\cov}[2]{\mathrm{Cov}\!\del{#1,#2}}
\newcommand{\tr}[1]{{\rm{tr}}\!\del{#1}}

\begin{document}
\title{Multidimensional stochastic liquidity \\ in Kyle's model of informed trading}
\author[1]{Ibrahim Ekren\thanks{Partially supported by NSF grant DMS-2406240}}
\author[1]{Evangelos A.\ Nikitopoulos}
\author[2]{Lu Vy}
\affil[1]{Department of Mathematics, University of Michigan\protect\\
\noindent 530 Church Street, Ann Arbor, MI 48109-1043 (USA)\vspace{2mm}}
\affil[2]{Department of Operations Research and\protect\\
Financial Engineering, Princeton University\protect\\
\noindent 98 Charlton Street, Princeton, NJ 08540 (USA)\protect\\
{Email: \tt \href{mailto:lv4809@princeton.edu}{lv4809@princeton.edu}}}
\date{\vspace{-5ex}}

\maketitle

\begin{abstract}
We develop a variational formulation of Kyle's model of informed trading that accommodates stochastic liquidity and multiple traded assets. The main equilibrium result is stated first: under a martingale dual condition, a matrix-valued martingale depth process generates a linear-Gaussian equilibrium with stochastic matrix-valued price impact. We derive this martingale from a primal--dual problem, inspired by causal optimal transport, that characterizes the endogenous speed at which the insider injects private information into prices; in general, this problem admits only local martingale optimizers, and the martingale dual condition is the hypothesis that the optimizer is a true martingale. We interpret informed trading as the optimal liquidation of private information and verify the construction in the scalar and common-eigenbasis cases. The fully general matrix-valued case reduces to a coupled matrix FBSDE, which we isolate as the remaining obstruction. Along the way, we establish an independently interesting Doob--Meyer decomposition for general (not necessarily symmetric) matrix-valued submartingales.

$\,$

\noindent\textbf{Keywords:} Kyle model, stochastic liquidity, insider trading, market depth, causal optimal transport, matrix-valued martingales

$\,$

\noindent\textbf{MSC (2020):} 91G15, 91G80, 60G44, 60H30, 49N90
\end{abstract}

\setcounter{tocdepth}{2}
{\setlength{\parskip}{0pt}\tableofcontents}

\section{Introduction}
Kyle--Back models study how private information is incorporated into prices when an informed trader
strategically trades against competitive market makers. In Kyle's original Gaussian model \cite{kyle},
market makers observe only aggregate order flow, prices respond linearly to order flow, market depth
is constant, and all private information is revealed by the terminal date. Back \cite{back1992}
extends this continuous-time equilibrium beyond Gaussian terminal values. Collin-Dufresne and Fos
\cite{cdf} introduce stochastic liquidity in the noise-trading volatility, showing that insiders time
their trades to periods of high noise volume and low price impact. A central feature of their
analysis is that the rate at which the insider releases information into prices cannot be specified
exogenously: it is determined endogenously in equilibrium, as the insider trades off the gains from
exploiting favorable liquidity states against the information revealed by doing so. Ekren, Mostowski, and \v{Z}itkovi\'c
\cite{emz} then combine stochastic liquidity with non-Gaussian terminal values through a nonlinear
state-variable construction.

This paper develops a multidimensional Gaussian theory with stochastic liquidity, organized around a
variational problem that makes this endogenous speed of information injection its central object.
Inspired by a causal optimal transport--type problem (as explained in appendix~\ref{tBdCdebe09}), the variational problem optimizes directly over
the rate at which the market maker's posterior covariance is driven down---equivalently, the rate at
which the insider injects private information into prices. Its solution \emph{is} the equilibrium
information-release speed, recovering as a derived quantity what \cite{cdf} identified as necessarily
endogenous, now in a multi-asset setting and through a transport-theoretic lens.

The dual of this variational problem produces a matrix-valued process $\bM^*$, interpreted as market
depth, whose inverse $\bfLambda^*=(\bM^*)^{-1}$ is the equilibrium price-impact matrix. A structural
subtlety drives the entire analysis: the optimizer of the variational problem
$\bM^*$ is a $\sF^{\bW}$-local martingale but \emph{not}, in general, a true martingale, and the
compactness argument establishing the existence of an optimizer delivers only a local martingale. Our
main result is therefore conditional: \emph{if} this dual optimizer is a true
martingale---a hypothesis we isolate as the martingale dual condition (MDC)---then it generates an
equilibrium in the multidimensional Kyle model with stochastic liquidity. The difficulty is intrinsic to multidimensionality: the matrices $\bC$, $\bfsigma_t$, and $\bM_t^*$
do not commute, so they cannot be simultaneously diagonalized, and the scalar arguments have no
analog. In one dimension MDC is established in \cite{emz} through BSDE arguments and a comparison
principle that signs the relevant scalar process; in the matrix setting, neither the BSDE reduction
nor the comparison principle has a counterpart. For this reason, we isolate MDC as a hypothesis.

The paper is organized around the equilibrium rather than around the derivation. Section~\ref{sec:main-result} states the main theorem first. Under MDC, the pricing rule
\[
\bP_t^*=\bp_0+\int_0^t(\bM_s^*)^{-1}\,\dif\bY_s^*
\]
and the corresponding absolutely continuous trading strategy form an equilibrium. The theorem also records the filtering covariance, the Gaussian bridge property, full revelation at maturity, and inconspicuousness. The proof is given in appendix~\ref{app:equilibrium-proofs}; there, the Kalman--Bucy filtering identities, the Gaussian bridge property, and the wealth decomposition are assembled to verify both rationality and optimality.

Section~\ref{sec:variational-construction} derives $\bM^*$ from the primal--dual variational problem over decreasing covariance paths. The primal problem treats the market maker's posterior covariance as the state variable, and its control is the speed of information release; informed trading is therefore an optimal liquidation problem in which the inventory being liquidated is private information. The dual problem identifies the martingale depth process as the shadow price of this information inventory, and it is here that the local- versus true-martingale distinction enters: the dual optimizer exists as a local martingale unconditionally, and MDC is precisely the statement that it is a true, interior martingale. The duality and sensitivity arguments are proved in appendix~\ref{app:dual-proofs}. The compactness step in the dual-existence proof requires a Doob--Meyer decomposition theorem for matrix-valued submartingales, which we establish for general, not necessarily symmetric, matrix-valued local submartingales in appendix~\ref{app:matrix-doob-meyer}. Since the theory of matrix-valued martingales seems relatively underdeveloped \cite{wangramdas}, the decomposition theorem is of independent interest.

Section~\ref{x854swjFhd} verifies the construction in the benchmark cases. The scalar setting recovers Kyle
\cite{kyle}, Back \cite{back1992}, Back--Pedersen \cite{bp}, Collin-Dufresne--Fos \cite{cdf}, and Ekren--Mostowski--\v{Z}itkovi\'c \cite{emz}. The multidimensional common-eigenbasis case recovers the constant-volatility multi-asset model of Back--Cocquemas--Ekren--Lioui \cite{cel2020} and gives a tractable class where MDC can be checked coordinatewise. Without a common time-independent eigenbasis, the verification of MDC reduces to the analysis of a fully coupled matrix FBSDE, which we leave open. The computations supporting the constant-volatility
and common-eigenbasis examples are collected in appendix~\ref{app:verification-proofs}.

\subsection{Notations}
We denote by $\mathcal{S}^n_{+}$ the set of $n\times n$ (symmetric) positive semi-definite matrices and by 
$\mathcal{S}^n_{++}\subseteq \mathcal{S}^n_{+}$ the subset of positive-definite matrices. We likewise write 
$\mathcal{S}^n_{-}$ and $\mathcal{S}^n_{--}$ for their negative semi-definite and negative definite counterparts.

Matrix inequalities are understood in the sense of the operator order. For $n \times n$ matrices $\bA$ and $\bB$, we write $\bA\leq\bB$ if $\bB-\bA\in\mathcal{S}^n_{+}$ and $\bA<\bB$ if $\bB-\bA\in\mathcal{S}^n_{++}$.

Boldface lowercase letters such as $\bx$ and $\by$ denote vectors, while boldface uppercase letters such as $\bM$ and $\bfLambda$ denote matrices. An exception is made for random vectors, which are written as italicized, boldface capital letters such as $\bX$ and $\bY$.

\section{The model}

In this section, we set up our model of interest.
We begin with an informal description and then move on to give precise definitions of all the objects appearing therein.

\subsection{Informal description}

A public announcement at time $T>0$ will reveal the post-announcement values of $n$ risky assets,
\[
\vt=(\tilde v_1,\dots,\tilde v_n)\in\dR^n.
\]
The law of $\vt$ is common knowledge, but its realized value is known from time $0$ only to a single informed trader.
At time $0$, the assets trade at some initial price $\bP_0\in\dR^n$.

Over $[0,T]$, the informed trader submits orders $\dif\bX_t\in\dR^n$ and accumulates a position $\bX_t\in\dR^n$.
Her trades are mixed with noise (liquidity) trades
\[
\dif\bZ_t=\bfsigma_t\dif\bB_t,
\]
where $\bB$ is an $n$-dimensional Brownian motion and the (possibly stochastic) volatility matrix
$\bfsigma_t\in\cS^n_{++}$ is driven by a second Brownian motion $\bW$, independent of $\bB$.
A competitive market maker observes the aggregate order flow
\[
\bY_t=\bX_t+\bZ_t
\]
together with the public Brownian motion $\bW$ driving the noise volatility.
Using the information generated by $(\bY,\bW)$, he sets prices by the rational-expectations rule
\[
\bP_t=\E{\vt \,\middle|\, \sF_t^{\bY,\bW}},\qquad 0 \leq t \leq T.
\]

\subsection{Probabilistic set-up and definitions}

Let $(\Omega,\sF,(\sF_t)_{0 \leq t \leq T},\dP)$ be a filtered probability space satisfying the usual conditions.
Let $\bB=(\bB_t)_{0\le t\le T}$ and $\bW=(\bW_t)_{0\le t\le T}$ be independent $\dR^n$-valued Brownian motions on this space.
Let $\bfsigma=(\bfsigma_t)_{0 \leq t \leq T}$ be an $\cS_{++}^n$-valued, $(\sF_t^{\bW})_{0 \leq t \leq T}$-adapted process satisfying
\[
\int_0^T \E{\|\bfsigma_t\|^2}\dif t < \infty.
\]
We call $\bfsigma$ the \emph{noise volatility matrix}, and define the \emph{noise demand} by
\[
\bZ_t:=\int_0^t \bfsigma_s\dif\bB_s, \qquad 0 \le t \le T.
\]

Let $\vt=(\tilde v_1,\dots,\tilde v_n)$ be an $\sF_T$-measurable random vector with
\[
\vt\sim\cN(\bp_0,\bC),
\qquad \bC\in\cS^n_{++}.
\]
We assume that $\vt$ is independent of $\sF_T^{\bB,\bW}$.
The random vector $\vt$ is the \emph{terminal fundamental value}.

We write
\[
\sF_t^I:=\sF_t^{\bB,\bW}\vee\sigma(\vt), \qquad 0 \le t \le T,
\]
for the insider's filtration. All filtrations are understood to be augmented with $\dP$-null sets. All semimartingales appearing in the model are assumed continuous unless explicitly stated otherwise.

\begin{definition}[Trading strategy]
A \emph{trading strategy} is an $\dR^n$-valued $\sF^I$-semimartingale $\bX=(\bX_t)_{0\le t\le T}$ with $\bX_0=\mathbf 0$. Given a trading strategy $\bX$, the associated \emph{aggregate order flow} is
\[
\bY^{\bX}:=\bX+\bZ.
\]
The market maker observes $\bY^{\bX}$ and $\bW$. Accordingly, his filtration is
\[
\sF_t^{M,\bX}:=\sF_t^{\bY^{\bX},\bW}
=\sigma\big((\bY_s^{\bX},\bW_s)_{0\le s\le t}\big),
\qquad 0 \le t \le T.
\]
\end{definition}

\begin{definition}[Pricing rule]
\label{def:pricingrule}
A \emph{pricing rule} is a nonanticipative map
\[
\bfP:\bY\to \bfP(\bY),
\]
which assigns to each aggregate order flow $\bY$ an $\dR^n$-valued $\sF^{\bY,\bW}$-semimartingale $\bfP(\bY)$.

In this paper, we restrict attention to pricing rules of the form
\[
\bfP(\bY)_t=\bfH\bigl(t,\bfxi_t^{\bY}\bigr),
\qquad
\bfxi_t^{\bY}:=\int_0^t \bfLambda_s\dif\bY_s,
\qquad 0 \le t \le T,
\]
where $(\bfH,\bfLambda)$ satisfies the following:
\begin{enumerate}[label=(\roman*),font=\normalfont]
\item $\bfH:[0,T]\times\dR^n\to\dR^n$ is of class $C^{1,2}$, and for each $t\in[0,T]$,
the map $\xi\mapsto \bfH(t,\xi)$ is the gradient of a convex function
$\varphi(t,\cdot)$;
\item $\bfLambda$ is an $\cS^n_{++}$-valued $\sF^{\bW}$-semimartingale;
\item if
\[
\bfxi_t^{\,0}:=\int_0^t \bfLambda_s\dif\bZ_s, \qquad 0 \le t \le T,
\]
then
\[
\E{\|\bfH(T,\bfxi_T^{\,0})\|^2+\int_0^T \|\bfH(t,\bfxi_t^{\,0})\|^2\dif t}<\infty.
\]
\end{enumerate}
\end{definition}

Condition (i) is a multi-asset analog of monotonicity: higher demand should not be bad news. Condition (iii) is an integrability condition on the price process generated by pure noise demand. Doubling-type pathologies are ruled out below through the admissibility condition on wealth.

\begin{remark}
It is important to distinguish between the pricing rule $\bfP$, which is a functional of the order flow,
and the realized price process induced by a particular strategy $\bX$, namely
\[
\bP^{\bX}:=\bfP(\bY^{\bX}).
\]
Likewise, for a given strategy $\bX$, we write
\[
\bfxi_t^{\bX}:=\bfxi_t^{\bY^{\bX}}
=\int_0^t \bfLambda_s\dif\bY_s^{\bX},
\qquad
\bP_t^{\bX}=\bfH\bigl(t,\bfxi_t^{\bX}\bigr).
\]
\end{remark}

For two continuous $\dR^n$-valued semimartingales $\bX$ and $\bP$, we write
\[
\langle \bX,\bP \rangle_t:=\sum_{i=1}^n \langle X^i,P^i\rangle_t,
\qquad 0\le t\le T,
\]
for their scalar quadratic covariation.

\begin{definition}[Wealth]
Let $\bX$ be a trading strategy and let $\bP$ be an $\dR^n$-valued semimartingale.
The \emph{insider's wealth process} associated with $(\bX,\bP)$ is
\[
\Pi_t(\bX,\bP)
:=
\int_0^t (\vt-\bP_s)^{\top}\dif\bX_s
-
\langle \bX,\bP \rangle_t,
\qquad 0\le t\le T.
\]
Equivalently, by integration by parts and $\bX_0=\mathbf 0$,
\[
\Pi_t(\bX,\bP)
=
(\vt-\bP_t)^{\top}\bX_t
+
\int_0^t \bX_{s}^{\top}\dif\bP_s.
\]
\end{definition}

\begin{definition}[Admissible trading strategy]
\label{def:admissible}
Fix a pricing rule $\bfP$.
A trading strategy $\bX$ is said to be \emph{admissible for $\bfP$} if the realized price process
\[
\bP^{\bX}:=\bfP(\bY^{\bX})=\bfP(\bX+\bZ)
\]
is well-defined and the wealth process
\[
\Pi_t(\bX,\bP^{\bX}),\qquad 0\le t\le T,
\]
is uniformly bounded from below by an integrable random variable. That is, there exists an integrable random variable \(K^\bX\) such that
\[
\Pi_t(\bX,\bP^{\bX})\ge -K^\bX,
\qquad 0\le t\le T.
\]
\end{definition}

The role of the state variable $\bfxi^{\bX}$ is to compress the history
$(\bY_s^{\bX})_{0\le s\le t}$ into a tractable, typically Markovian, state variable.
The process $\bfLambda$ plays the role of Kyle's price impact, and $\bfLambda^{-1}$ corresponds to market depth.
In equilibrium, one obtains the linear-innovation form
\[
\bP_t^{\bX}=\bp_0+\int_0^t \bfLambda_s\dif\bY_s^{\bX}.
\]

\begin{definition}[Optimal strategy]
\label{def:optimal}
Fix a pricing rule $\bfP$.
A $\bfP$-admissible trading strategy $\bX^*$ is \emph{optimal for $\bfP$} if for every
$\bfP$-admissible trading strategy $\bX$ and every $\bv\in\operatorname{supp}(\vt)$,
\[
\E{\Pi_T\bigl(\bX^*,\bP^{\bX^*}\bigr)\,\middle| \vt=\bv}
\ge
\E{\Pi_T\bigl(\bX,\bP^{\bX}\bigr)\,\middle| \vt=\bv},
\]
where
\[
\bP^{\bX^*}:=\bfP(\bX^*+\bZ),
\qquad
\bP^{\bX}:=\bfP(\bX+\bZ).
\]
\end{definition}

\begin{definition}[Rational pricing rule]
\label{def:rational}
Let $\bX$ be a trading strategy.
A pricing rule $\bfP$ is \emph{rational for $\bX$} if
\[
\bP_t^{\bX}
=
\E{\vt\,\middle| \sF_t^{M,\bX}},
\qquad 0 \le t \le T,
\]
that is,
\[
\bfP(\bX+\bZ)_t
=
\E{\vt\,\middle| \sF_t^{\bX+\bZ,\bW}}.
\]
Equivalently, when $\bfP$ is generated by $(\bfH,\bfLambda)$,
\[
\bfH(t,\bfxi_t^{\bX})
=
\E{\vt\middle| \sF_t^{M,\bX}},
\qquad 0 \le t \le T.
\]
\end{definition}

\begin{definition}[Equilibrium]
\label{def:equilibrium}
A pair $(\bfP,\bX^*)$ consisting of a pricing rule $\bfP$ and a trading strategy $\bX^*$ is an \emph{equilibrium} if:
\begin{enumerate}[label=(\roman*),font=\normalfont]
\item (\emph{Optimality}) $\bX^*$ is admissible for $\bfP$, and for every $\bfP$-admissible trading strategy $\bX$ and every $\bv\in\operatorname{supp}(\vt)$,
\[
\E{\Pi_T\bigl(\bX,\bfP(\bX+\bZ)\bigr)\,\middle|\,\vt=\bv}
\le
\E{\Pi_T\bigl(\bX^*,\bfP(\bX^*+\bZ)\bigr)\,\middle|\,\vt=\bv}.
\]
\item (\emph{Rationality}) The realized price process $\bP^{\bX^*}:=\bfP(\bX^*+\bZ)$ satisfies
\[
\bP_t^{\bX^*}=\E{\vt\,\middle|\,\sF_t^{\bX^*+\bZ,\,\bW}},\qquad 0\le t\le T.
\]
\end{enumerate}
When $\bfP$ is generated by a pair $(\bfH,\bfLambda)$ as in Definition~\ref{def:pricingrule}, we also say that the triple $(\bX^*,\bfH,\bfLambda)$ is an \emph{equilibrium}.
\end{definition}

\begin{remark}
The pricing rule $\bfP$ is a single functional, fixed throughout. The optimization in (i) varies the strategy $\bX$ while keeping $\bfP$ fixed: different strategies produce different realized order flows $\bX+\bZ$ and hence different realized price processes $\bfP(\bX+\bZ)$, all generated by the same functional. The rationality condition (ii) is imposed only along the equilibrium realized path $\bX^*+\bZ$, not off-equilibrium.
\end{remark}

\subsection{Summary of observability}

The key informational asymmetry is that the insider observes the terminal value $\vt$
and chooses her trading strategy accordingly, whereas the market maker does not observe
the decomposition
\[
\bY_t=\bX_t+\bZ_t
\]
into informed and noise demand. He observes only the aggregate order flow $\bY$
together with the public signal $\bW$.
Thus, $\bX$ and $\bZ$ are hidden components, while $\bY$, the induced state variable,
and the induced price process are observable.

In our formulation, the market maker precommits to a pricing rule, namely a functional
of the observable order flow. In the class studied here, this pricing rule is generated
by a pair $(\bfH,\bfLambda)$ through
\[
\bfxi_t^{\bY}=\int_0^t \bfLambda_s\,\dif \bY_s,
\qquad
\bfP(\bY)_t=\bfH(t,\bfxi_t^{\bY}).
\]
For the realized aggregate order flow $\bY^{\bX}=\bX+\bZ$, the resulting price process is
\[
\bP_t^{\bX}=\bfP(\bY^{\bX})_t=\bfH(t,\bfxi_t^{\bX}).
\]

\begin{table}[H]
\centering
\renewcommand{\arraystretch}{1.15}
\setlength{\tabcolsep}{5pt}
\begin{tabular}{c|c|c|c|c}
Symbol & Meaning & Relation(s) & Insider & Market maker \\
\hline
$\vt$ & terminal fundamental value & $\vt\sim\cN(\bp_0,\bC)$ & yes & no \\
$F$ & law of $\vt$ & common prior & yes & yes \\
$\bB_t$ & noise Brownian motion & $\bZ_t=\int_0^t \bfsigma_s\,\dif\bB_s$ & yes & no \\
$\bW_t$ & public Brownian signal & public factor & yes & yes \\
$\bfsigma_t$ & noise volatility matrix & drives $\bZ$ & yes & yes \\
$\bX_t$ & insider's cumulative demand & strategy chosen using $\vt$ & yes & no \\
$\bZ_t$ & noise demand & $\dif\bZ_t=\bfsigma_t\,\dif\bB_t$ & yes & no \\
$\bY_t$ & aggregate order flow & $\bY_t=\bX_t+\bZ_t$ & yes & yes \\
$\bfxi_t^{\bY}$ & state variable & $\bfxi_t^{\bY}=\int_0^t \bfLambda_s\,\dif\bY_s$ & yes & yes \\
$\bfP$ & pricing rule & $\bfP(\bY)_t=\bfH(t,\bfxi_t^{\bY})$ & yes & yes \\
$\bP_t^{\bX}$ & realized price process & $\bP_t^{\bX}=\bfP(\bX+\bZ)_t$ & yes & yes
\end{tabular}
\end{table}

Equivalently, if one works directly with the pair $(\bfH,\bfLambda)$ rather than the
functional $\bfP$, then the market maker observes $\bY$ and forms prices through the
observable state variable $\bfxi^{\bY}$. What remains hidden from him is the split
between informed trading $\bX$ and noise trading $\bZ$.

\section{Main result}
\label{sec:main-result}

The equilibrium is governed by a matrix-valued martingale which plays the role of market depth. We state the result before deriving this martingale from the variational problem.

Let
\begin{equation}
\label{2ikD666Yhq}
\cM
=
\left\{
\bM_{[0,T]}:
\begin{array}{l}
\bM_t=\E{\bM_T\mid\sF_t^{\bW}}\ \text{for some }
\sF_T^{\bW}\text{-measurable}\\
\bM_T\in L^1(\Omega;\cS^n_{++}), \; \bM_t\in\cS^n_{++}\ \text{for }0\le t\le T
\end{array}
\right\}.
\end{equation}

\begin{definition}[Martingale dual condition]
\label{def:MDC}
We say that the data $(\bC,(\bfsigma_t)_{0\le t\le T})$ satisfy the
\emph{martingale dual condition} MDC if the minimization problem
\begin{equation}
\label{RUHFk3mIp0}
\inf_{\bM\in\cM}
\left\{
\frac12\tr{\bC\bM_0}
+
\frac12\E{\int_0^T \tr{\bfsigma_t^2\bM_t^{-1}}\,\dif t}
\right\}
\end{equation}
admits an optimizer $\bM^*\in\cM$ which is an interior minimizer in the
following sense: for every bounded symmetric $\sF^{\bW}$-martingale
$\bN=(\bN_t)_{0\le t\le T}$, there exists $\varepsilon_0>0$ such that
\[
    \bM_t^*+\varepsilon \bN_t\in\cS^n_{++},
    \qquad 0\le t\le T,
\]
for all $|\varepsilon|<\varepsilon_0$.
\end{definition}
\begin{remark}
Note that the problem \eqref{RUHFk3mIp0} is convex on the convex set $\cM$, so an optimizing sequence always exists, and (by the compactness argument of Theorem~\ref{hodau0V7QJ}) converges to a limiting supermartingale. We are, however, unable to show that this limit is a true martingale in general. We show below that many models in the literature do satisfy the martingale dual~condition.
\end{remark}
\begin{theorem}[Equilibrium under the martingale dual condition]
\label{I3YrvmnwLc}
Suppose that the data $(\bC,(\bfsigma_t)_{0\le t\le T})$ satisfy MDC, and let $(\bM_t^*)_{0\le t\le T}$ be the resulting martingale optimizer. Define
\begin{align}
\label{CrLPP8jKQG}
\bfLambda_t^* &= (\bM_t^*)^{-1}, \\
\label{2vGs8ogSts}
\bfH^*(t,\bfxi)&=\bp_0+\bfxi,
\end{align}
and, for the realized aggregate order flow,
\begin{equation}
\label{eq:xi-star-main}
\bfxi_t^*:=\int_0^t\bfLambda_s^*\,\dif\bY_s^*.
\end{equation}
Define the covariance path
\begin{equation}
\label{eq:Sigma-star-def}
\bfSigma_t^*
:=
\bC-\int_0^t(\bM_s^*)^{-1}\bfsigma_s^2(\bM_s^*)^{-1}\,\dif s,
\qquad 0\le t\le T,
\end{equation}
which is well-defined with $\bfSigma_0^*=\bC$. Under MDC, Proposition~\ref{prop:MDC-terminal-covariance} shows that $\bfSigma_t^*$ is positive definite for $t<T$---so that the integrand in the strategy below is well-defined---and admits the equivalent backward representation
\begin{equation}
\label{eq:Sigma-star-backward}
\bfSigma_t^*=\int_t^T(\bM_s^*)^{-1}\bfsigma_s^2(\bM_s^*)^{-1}\,\dif s,
\end{equation}
so that in particular $\bfSigma_T^*={\bf 0}$.
Then the trading strategy
\begin{equation}
\label{ud3k6AlV7q}
\bX_t^*
=
\int_0^t
\bfsigma_s^2(\bM_s^*)^{-1}(\bfSigma_s^*)^{-1}
\del{\vt-\bp_0-\bfxi_s^*}\,\dif s
\end{equation}
together with the pricing rule $\bfP^*(\bY)_t=\bfH^*(t,\int_0^t\bfLambda_s^*\,\dif\bY_s)$ is an equilibrium.

Moreover, along the equilibrium path:
\begin{enumerate}[label=(\roman*),font=\normalfont]
\item the price process is
\begin{equation}
\bP_t^*=\bp_0+\bfxi_t^*
=\bp_0+\int_0^t(\bM_s^*)^{-1}\,\dif\bY_s^*
=\E{\vt\,\middle|\sF_t^M};
\end{equation}
\item the market maker's posterior covariance is
\begin{equation}
{\rm Var}\!\del{\vt\,\middle|\sF_t^M}=\bfSigma_t^*;
\end{equation}
\item conditionally on $\sF_t^M$,
\begin{equation}
\label{eq:conditional-gaussian-main}
\bfxi_T^*\sim \cN(\bfxi_t^*,\bfSigma_t^*),
\qquad
\bP_T^*=\vt\quad \text{a.s.};
\end{equation}
\item the strategy is inconspicuous:
\begin{equation}
\label{eq:inconspicuous-main}
\E{\frac{\dif\bX_t^*}{\dif t}\middle|\sF_t^M}=0
\qquad\text{for a.e. }t.
\end{equation}
\end{enumerate}
\end{theorem}

The proof has three ingredients. First, the Kalman--Bucy filter identifies \eqref{SWjcJYb4sK} and \eqref{ENZWGUDa1A}. Second, the terminal covariance identity implied by MDC, proved in Proposition~\ref{prop:MDC-terminal-covariance}, turns the state process into a Gaussian bridge ending at $\vt-\bp_0$. Third, a quadratic wealth decomposition verifies that no admissible strategy can outperform \eqref{ud3k6AlV7q}. The details are collected in appendix~\ref{app:equilibrium-proofs}.

\begin{proposition}[Insider value and optimality]
\label{prop:insider-value}
Under MDC, if the market maker uses the pricing rule \eqref{CrLPP8jKQG}--\eqref{2vGs8ogSts}, then among all $\bfP^*$-admissible strategies the candidate strategy \eqref{ud3k6AlV7q} maximizes the insider's $\vt$-conditional expected profit
\[
\E{\int_0^T(\vt-\bP_t)^{\top}\dif\bX_t - \langle\bX,\bP\rangle_T \,\middle|\, \sF^I_0}.
\]
The maximum is attained precisely by strategies that are inconspicuous, $\E{\dif\bX_t/\dif t\mid\sF^M_t}={\bf 0}$, and satisfy $\bfxi_T=\vt-\bp_0$ almost surely.
\end{proposition}

The candidate strategy thus releases information so as to be statistically invisible to the market maker at each instant, yet it is profit-maximizing for the informed trader precisely because it attains this bound. The explicit value of the maximal expected profit, together with the wealth decomposition that establishes it, is given in appendix~\ref{app:equilibrium-proofs}; see Theorem~\ref{thm:insider-value}. This records the insider's value content in the main result, as in \cite{cdf} and \cite{emz}, while deferring the verification.

\section{The variational construction}
\label{sec:variational-construction}
In this section we derive the dual problem \eqref{RUHFk3mIp0} from a primal variational
problem, and identify the matrix-valued martingale $\bM^*$ of section~\ref{sec:main-result} as the
multiplier enforcing the terminal constraint. The primal problem records the market maker's
perspective: its state variable is the posterior covariance of $\vt$, and its control is the
rate at which that covariance is driven down to zero by the release of information into prices.
This is the precise sense, made quantitative in section~\ref{sec:economic-interpretation}, in which
informed trading is the optimal liquidation of an information inventory.

\subsection{The primal problem}
\label{hxjKCTerXp}
Fix $\bC\in\cS^n_{++}$ and consider
\begin{equation}
\label{OxpDCjjac4}
\begin{aligned}
\sup_{\bfSigma}
\quad & \E{\int_0^T \tr{\sqrt{-\,\bfsigma_t\,\dot{\bfSigma}_t\,\bfsigma_t}}\dif t} \\
\text{s.t.}\quad
& \bfSigma_0=\bC,\qquad \bfSigma_T\geq{\bf 0},\\
& \bfSigma \ \text{is absolutely continuous and } \dot{\bfSigma}_t\in \cS^n_{-}\ \text{for a.e.\ } t\in[0,T].
\end{aligned}
\end{equation}
Here the variable $(\bfSigma_t)_{0\le t\le T}$ is the market maker's posterior covariance
${\rm Var}(\vt\mid\sF^M_t)$, starting from the prior $\bfSigma_0=\bC$. The constraint
$\dot{\bfSigma}_t\in\cS^n_{-}$ is monotonicity---posterior covariance can only decrease---so the
control is the instantaneous information-release rate $-\dot{\bfSigma}_t\in\cS^n_{+}$, and the
integrand $\tr{\sqrt{-\bfsigma_t\dot{\bfSigma}_t\bfsigma_t}}$ is the value of releasing at this rate
against noise of volatility $\bfsigma_t$ (reducing to $\sigma_t\sqrt{-\dot\Sigma_t}$ in the scalar
case). The terminal requirement is the inequality $\bfSigma_T\geq{\bf 0}$, the form dual to the
positive-definite cone $\cM$. Appendix~\ref{tBdCdebe09} gives a heuristic derivation of
\eqref{OxpDCjjac4} from a causal optimal transport--type problem.

Accordingly, we first define the
class of velocities subject only to the pointwise constraint, together with the integrability needed
for the objective and the multiplier pairing to be well~defined:
\begin{equation}
\label{def:A0}
\cA
:=
\left\{
\dot{\bfSigma} :
\begin{array}{l}
\dot{\bfSigma}_t\in \cS^n_{-}\ \text{for a.e. }t,\\
\dot{\bfSigma}\ \text{is } \sF^{\bW}\text{-adapted with }
\E{\int_0^T \|\dot{\bfSigma}_t\|\,\dif t} <\infty,\\
\displaystyle
\E{\int_0^T \tr{\sqrt{-\bfsigma_t\dot{\bfSigma}_t\bfsigma_t}}\,\dif t} <\infty
\end{array}
\right\}.
\end{equation}
The three conditions defining $\cA$ are, in order: the pointwise monotonicity constraint inherited
from \eqref{OxpDCjjac4}; adaptedness of the control to the public filtration $\sF^{\bW}$, with an
$L^1$ bound ensuring the running integral $\int_0^t\dot{\bfSigma}_s\,\dif s$ is well defined and the
multiplier pairing below is finite; and finiteness of the objective itself, which rules out
velocities that release information arbitrarily~fast.

For a fixed initial covariance $\bC\in\cS^n_{++}$, we define the terminally
admissible subclass
\begin{equation}
\label{def:AC}
\cA(\bC)
:=
\left\{
\dot{\bfSigma}\in\cA:
\bC+\int_0^T \dot{\bfSigma}_t\,\dif t\ \geq\ \mathbf 0
\right\}.
\end{equation}
Given $\dot{\bfSigma}\in\cA(\bC)$, the associated covariance path is
\[
\bfSigma_t
=
\bC+\int_0^t\dot{\bfSigma}_s\,\dif s,
\qquad 0\le t\le T.
\]
Thus the primal problem \eqref{OxpDCjjac4} can equivalently be written as
\[
\sup_{\dot{\bfSigma}\in\cA(\bC)}
\E{\int_0^T
\tr{\sqrt{-\bfsigma_t\dot{\bfSigma}_t\bfsigma_t}}\,\dif t}.
\]

\subsection{The dual problem}

The terminal constraint in \eqref{def:AC} is imposed by the matrix-valued martingale multiplier $\bM\in\cM$ defined in \eqref{2ikD666Yhq}. Since $\dot{\bfSigma}$ is $\sF^{\bW}$-adapted, optional projection gives
\[
\E{\tr{\bM_T\int_0^T\dot{\bfSigma}_t\,\dif t}}
=
\E{\int_0^T\tr{\bM_t\dot{\bfSigma}_t}\,\dif t}.
\]

\begin{definition}[Lagrangian]
Let $\cA$ be the terminally-unconstrained velocity class in \eqref{def:A0}, and let $\cM$ be as in \eqref{2ikD666Yhq}. The \emph{Lagrangian}
$L:\cA\times\cM\to\dR\cup\{-\infty\}$ is
\begin{equation}
\label{PC4t8KCzcX}
L(\dot{\bfSigma}_{[0,T]},\bM_{[0,T]})
=
\frac{1}{2}\tr{\bC\bM_0}
+
\E{\int_0^T
\tr{
\sqrt{-\bfsigma_t\dot{\bfSigma}_t\bfsigma_t}
+
\frac{1}{2}\bM_t\dot{\bfSigma}_t
}\,\dif t}.
\end{equation}
\end{definition}

\begin{proposition}
\label{wn0NtKxptx}
The dual function is
\begin{equation}
\label{YTF6qPCOFy}
g\!\del{\bM_{[0,T]}} := \sup_{\dot{\bfSigma}\in\cA} L(\dot{\bfSigma}_{[0,T]},\bM_{[0,T]}) = \frac{1}{2}\tr{\bC \bM_0} + \frac{1}{2}\E{\int_0^T \tr{\bfsigma^2_t \bM_t^{-1}}\dif t}.
\end{equation}
\end{proposition}

The proof is deferred to appendix~\ref{app:dual-proofs}. The quantity $\int_0^T \tr{\bfsigma_t^2\bM_t^{-1}}\,\dif t$ is the slippage-cost term appearing in Collin-Dufresne and Fos \cite{cdf}.

\begin{proposition}
\label{FOHnj1xlta}
The dual function $g$ is convex.
\end{proposition}

The proof is given in appendix~\ref{app:dual-proofs}.

\subsection{Duality and the martingale dual condition}

\begin{proposition}
\label{prop:MDC-terminal-covariance}
Suppose that $(\bC,(\bfsigma_t)_{0\le t\le T})$ satisfies MDC, and let
$\bM^*$ be the corresponding interior martingale minimizer. Then
\begin{equation}
\label{eq:MDC-terminal-constraint}
\bC=
\int_0^T
(\bM_t^*)^{-1}\bfsigma_t^2(\bM_t^*)^{-1}\,\dif t.
\end{equation}
Consequently, the covariance path
\begin{equation}
\label{eq:Sigma-star-main}
\bfSigma_t^*
:=
\bC-
\int_0^t
(\bM_s^*)^{-1}\bfsigma_s^2(\bM_s^*)^{-1}\,\dif s
=
\int_t^T
(\bM_s^*)^{-1}\bfsigma_s^2(\bM_s^*)^{-1}\,\dif s
\end{equation}
is well-defined, positive definite for $t<T$, and satisfies
$\bfSigma_0^*=\bC$ and $\bfSigma_T^*={\bf 0}$.
\end{proposition}

\begin{theorem}
\label{hodau0V7QJ}
The maximization problem \eqref{OxpDCjjac4} admits the dual formulation
\begin{equation}
\inf_{\bM\in\cM}g\!\del{\bM}
\end{equation}
where $g$ is given by \eqref{YTF6qPCOFy}. Furthermore, if $\bfsigma$ is uniformly bounded from above and below and $\bC$ is positive definite, there exists a unique $\sF^{\bW}$-local martingale $\bM^{*,0,\bC}$, possibly lying outside $\cM$, satisfying
\[
g(\bM^{*,0,\bC}) \;\leq\; \inf_{\bM\in\cM} g(\bM),
\qquad
|\bM^{*,0,\bC}_T|<\infty \;\text{a.s.}
\]
The condition that the local-martingale optimizer $\bM^{*,0,\bC}$ belongs to $\cM$ and is an interior minimizer is precisely MDC in Definition~\ref{def:MDC}. Under this condition, Proposition~\ref{prop:MDC-terminal-covariance} yields the terminal covariance identity \eqref{eq:MDC-terminal-constraint}.
\end{theorem}

The proof is given in appendix~\ref{app:dual-proofs}.

\begin{corollary}
\label{cor:strong-duality}
Suppose the data $(\bC,(\bfsigma_t)_{0\le t\le T})$ satisfies the martingale dual condition (Definition~\ref{def:MDC}), and let $\bM^*:=\bM^{*,0,\bC}\in\cM$ be the resulting $\sF^{\bW}$-martingale optimizer of \eqref{RUHFk3mIp0} furnished by Theorem~\ref{hodau0V7QJ}. Then strong duality holds in the sense that
\begin{equation}
\label{9CcesNg5w6}
\sup_{\dot{\bfSigma}}\inf_{\bM} L(\dot{\bfSigma}_{[0,T]},\bM_{[0,T]}) = \inf_{\bM}\sup_{\dot{\bfSigma}} L(\dot{\bfSigma}_{[0,T]},\bM_{[0,T]}),
\end{equation}
where the saddle point is given by $\bM^*$ and
\begin{equation}
\label{Rwq3zFmuoG}
\dot{\bfSigma}^*_t
=
-\sbr{\bM_t^*}^{-1}\bfsigma_t^2\sbr{\bM_t^*}^{-1},
\qquad 0\le t\le T,
\end{equation}
which satisfies
\[
\bC + \int_0^T \dot{\bfSigma}^*_t \dif t = {\bf 0}.
\]
\end{corollary}

The proof is given in appendix~\ref{app:dual-proofs}.

\subsection{Sensitivity of the dual value}

We now record a sensitivity property of the dual value with respect to the initial covariance, treated as a parameter. For $\Sigma\in\cS^{n}_{++}$, write
\begin{equation}
\label{eq:G-Sigma-def}
G_\Sigma\!\del{\bM_{[0,T]}} := \frac{1}{2}\tr{\Sigma\,\bM_0} + \frac{1}{2}\E{\int_0^T \tr{\bfsigma^2_s\,\bM_s^{-1}}\dif s}
\end{equation}
for the dual function with $\Sigma$ in place of $\bC$ in the boundary term, and
\begin{equation}
\label{eq:G-star-def}
G^*(\Sigma) := G_\Sigma\!\del{\bM^{*,0,\Sigma}}
\end{equation}
for the corresponding optimal dual value, where $\bM^{*,0,\Sigma}$ is the unique $\sF^{\bW}$-local martingale optimizer furnished by Theorem~\ref{hodau0V7QJ} at initial covariance $\Sigma$.

\begin{proposition}
\label{TUpmwPhTpr}
The map $G^*\colon\cS^{n}_{++}\to\dR$ is continuously differentiable, with
\begin{equation}
\label{eq:G-star-grad}
D_\Sigma G^*(\Sigma) = \frac{1}{2}\bM^{*,0,\Sigma}_0.
\end{equation}
Moreover, the map $\Sigma\mapsto\bM^{*,0,\Sigma}_0$ is continuous on $\cS^{n}_{++}$.
\end{proposition}

The proof is given in appendix~\ref{app:dual-proofs}.
\pagebreak

\begin{corollary}
\label{cor:V-diff}
For \(\bfSigma\in\cS^n_{++}\), let \(V(0,\bfSigma)\) denote the value of the primal problem
\eqref{OxpDCjjac4} with initial covariance \(\bfSigma\) in place of \(\bC\). Suppose the martingale
dual condition (Definition~\ref{def:MDC}) holds for the data
\((\bfSigma,(\bfsigma_t)_{0\le t\le T})\) for every \(\bfSigma\) in an open neighborhood
\(\cO\subset\cS^{n}_{++}\) of \(\bC\). Then \(V(0,\cdot)\) is continuously differentiable on
\(\cO\), with
\begin{equation}
D_{\bfSigma} V(0,\bfSigma)
=
\frac{1}{2}\bM^{*,0,\bfSigma}_0,
\qquad \bfSigma\in\cO.
\end{equation}
\end{corollary}

The proof is given in appendix~\ref{app:dual-proofs}.

\section{Verification and examples}
\label{x854swjFhd}

Theorem~\ref{I3YrvmnwLc} produces an equilibrium under MDC. We now identify the main settings in which the condition can be verified and recover familiar Kyle--Back equilibria as special~cases.

\subsection{Kyle 1985}
Fix $\sigma > 0$ constant and take $T = 1$. The primal problem \eqref{OxpDCjjac4} reads
\begin{equation}
\max_{\dot\Sigma \leq 0}\int_0^1 \sigma \sqrt{-\dot\Sigma(t)}\dif t,
\end{equation}
where the maximization is taken over all differentiable functions $\Sigma(t)$ satisfying $\Sigma(0)=\Sigma_{0}$ and $\Sigma(1)=0$.

We turn to standard calculus of variations. The optimal ${\Sigma}^*$ if it exists satisfies the Euler-Lagrange equation
\begin{equation}
\label{sPbl9R6sQb}
\del{-\dot{\Sigma}^*(t)}^{-1/2} = \rm{constant}
\end{equation}
and the particular solution which satisfies $\Sigma^*(0)=\Sigma_{0}$ and $\Sigma^*(1)=0$ is
\begin{equation}
\label{4QLrsXZhgR}
\Sigma^*(t) = \Sigma_{0}\del{1 - t}.
\end{equation}
By \eqref{BUSu1MNQ9m}, the solution to the dual problem \eqref{RUHFk3mIp0} is
\begin{equation}
\label{ZHLQzC93um}
M^*_t = \frac{\sigma}{\sqrt{\Sigma_0}}.
\end{equation}
The candidate $M^*$ is constant, hence a true martingale and an interior minimizer; since $\Sigma^*(1)=0$ by \eqref{4QLrsXZhgR}, the terminal covariance identity holds and so does MDC. Theorem~\ref{I3YrvmnwLc} may therefore be invoked. It tells us that \eqref{4QLrsXZhgR} is the filtering variance and that \eqref{ZHLQzC93um} is the market depth corresponding to some equilibrium. This equilibrium has the pricing rule
\begin{align}
\lambda^*_t &= \frac{\sqrt{\Sigma_0}}{\sigma} \\
H^*(t,\xi) &= p_0 + \xi
\end{align}
as given by \eqref{CrLPP8jKQG} and \eqref{2vGs8ogSts}. The price process \eqref{SWjcJYb4sK} is
$$
P^*_t = p_0 + \frac{\sqrt{\Sigma_0}}{\sigma} Y_t,
$$
and the insider's trading strategy \eqref{ud3k6AlV7q} is
$$
\dif X^*_t = \frac{\sigma}{\sqrt{\Sigma_0}}\frac{\tilde{v} - P^*_t}{1 - t} \dif t.
$$
Thus, we recover the equilibrium from \cite{kyle}.

\subsection{Back and Pedersen 1998, static information}
Let $\sigma(t)$ be time dependent but deterministic. The primal problem reads
\begin{equation}
\label{BAomThdW31}
\max_{\dot{\Sigma} \leq 0} \int_0^T\sigma(t)\sqrt{-\dot{\Sigma}(t)}\dif t.
\end{equation}
The optimal ${\Sigma}^*$ if it exists satisfies the Euler-Lagrange equation
\begin{equation}
\sigma(t)\del{-\dot\Sigma^*(t)}^{-1/2} = \rm{constant}.
\end{equation}
Imposing the conditions $\Sigma^*(0) = \Sigma_0$ and $\Sigma^*(T)=0$, the solution is
\begin{equation}
\Sigma^*(t) = \Sigma_0 \frac{\int_t^T \sigma(s)^{2}\dif s}{\int_0^T \sigma(s)^{2}\dif s}.
\end{equation}
Thus, the optimal $\dot\Sigma^*$ is
\begin{equation}
\label{GlDv15Ztsy}
\dot\Sigma^*(t) = -\Sigma_0 \frac{\sigma(t)^{2}}{\int_0^T\sigma(s)^{2}\dif s},
\end{equation}
and the optimal $M^*$ by \eqref{BUSu1MNQ9m} is
\begin{equation}
M^* = \sqrt{\frac{\int_0^T\sigma(s)^{2}\dif s}{\Sigma_0}}.
\end{equation}
Since $M^*$ is constant, it is a true martingale and an interior minimizer, and $\Sigma^*(T)=0$ gives the terminal covariance identity, so MDC holds; Theorem~\ref{I3YrvmnwLc} then asserts that an equilibrium exists. The equilibrium trading strategy \eqref{ud3k6AlV7q} is
\begin{equation}
\dif X_t^* = \sqrt{\frac{\int_0^T\sigma(s)^{2}\dif s}{\Sigma_0}}\sigma(t)^2 \frac{\tilde{v} - P^*_t}{\int_t^{T}\sigma(s)^{2} \dif s} \dif t
\end{equation}
and the equilibrium price process \eqref{SWjcJYb4sK} is
\begin{equation}
P^*_t = p_0 + \sqrt{\frac{\Sigma_0}{\int_0^T\sigma(s)^{2}\dif s}}Y^*_t.
\end{equation}
These coincide with equations (8) and (9) of \cite{bp} when $\sigma_S(t) = 0$ (in their notation). Thus we recover their equilibrium in the case of static information.

\subsection{Back, Cocquemas, Ekren, and Lioui 2020}
Back, Cocquemas, Ekren, and Lioui \cite{cel2020} construct a multidimensional Kyle--Back equilibrium through an optimal-transport formulation. In the risk-neutral Gaussian and constant-volatility specialization relevant here, their equilibrium has constant matrix price impact. We recover this benchmark by showing that the present variational problem selects the linear covariance path and a constant martingale depth.

Fix $\bfsigma > 0$ and take $T=1$. The primal problem \eqref{OxpDCjjac4} reads
\begin{equation}
\label{menT3otkFg}
\max_{\dot\bfSigma \leq 0}\int_0^T \tr{\sqrt{-\bfsigma \dot\bfSigma(t) \bfsigma}} \dif t
\end{equation}
where the maximization is taken over all differentiable matrix functions satisfying $\bfSigma(0)=\bC$ and $\bfSigma(T)=\bf{0}$.

\begin{proposition}
The solution to \eqref{menT3otkFg} is
\begin{equation}
\label{ui5XZD4G8F}
\bfSigma^*(t) = \bC\del{1 - \frac{t}{T}}.
\end{equation}
\end{proposition} 
Note that \eqref{ui5XZD4G8F} is the matrix analog of \eqref{4QLrsXZhgR}. We state it as a proposition because its proof may not be trivial to the reader.

The proof is deferred to appendix~\ref{app:verification-proofs}.
According to \eqref{1zwfZY3Cld}, the optimal $\bM^*$ is
$$
\bM^* = \bfsigma\del{\sqrt{\bfsigma\bC\bfsigma}}^{-1}\bfsigma.
$$
This is a constant (positive-definite) matrix, hence a true martingale and an interior minimizer; with $T=1$ the computation $\int_0^1(\bM^*)^{-1}\bfsigma^2(\bM^*)^{-1}\dif t=\bC$ gives the terminal covariance identity, so MDC holds. Theorem~\ref{I3YrvmnwLc} applies and we recover the equilibrium 
\begin{align*}
\bfH(t,\by) &= \bp_0 + \bfLambda\by \\
\dif\bX_t &= \frac{\bfLambda^{-1}\del{\vt - \bp_0}-\bY_t}{T-t}\dif t
\end{align*}
where $\bfLambda = \bfsigma^{-1}\del{\sqrt{\bfsigma\bC\bfsigma}}\bfsigma^{-1}$.

\subsection{Collin-Dufresne and Fos 2016, deterministic growth}
Let $\sigma_t$ be a one-dimensional diffusion with dynamics
\begin{equation}
\label{dpc20W9KFz}
\dif \sigma_t = \sigma_t m(t) \dif t + \sigma_t \nu(t,\sigma_t) \dif W_t,
\end{equation}
where $m(t)$ is deterministic.

We present two ways to recover this equilibrium, one by solving the primal problem and the other by solving the dual problem.

\subsubsection{The primal problem}
The primal problem reads
\begin{equation}
\label{dSsYjBUGdd}
\max_{\dot{\Sigma}<0}\E{\int_0^T\sigma_t\sqrt{-\dot\Sigma(t)}\dif t}.
\end{equation}
First note that $\sigma_t$ has the representation
\begin{equation}
\sigma_t = \sigma_0\exp\!\del{\int_0^t m(s) \dif s}\exp\!\del{-\int_0^t \frac{\nu(s,\sigma_s)^2}{2}\dif s + \int_0^t\nu(s,\sigma_s)\dif W_s}.
\end{equation}
This motivates the Radon Nikodym derivative
\begin{equation}
\frac{\dif\dQ}{\dif\dP} = \exp\!\del{-\int_0^T \frac{\nu(s,\sigma_s)^2}{2}\dif s + \int_0^T\nu(s,\sigma_s)\dif W_s},
\end{equation}
where $\dP$ is the probability measure in the expectation \eqref{dSsYjBUGdd}. We assume $\nu(s,\sigma_s)$ satisfies the Novikov condition
$$
\E{\exp\!\del{\frac{1}{2}\int_0^T \nu(s,\sigma_s)^2 \dif s}} < \infty.
$$
Once this assumption is granted, we may write \eqref{dSsYjBUGdd} as
\begin{equation}
\label{7xqVkCTVPT}
\max_{\dot{\Sigma}<0}\E{\int_0^T\sigma_t\sqrt{-\dot\Sigma(t)}\dif t} = \max_{\dot{\Sigma}<0}\mathbb{E}_{\dQ}\!\sbr{\int_0^T\sigma_0\exp\!\del{\int_0^t m(s)\dif s}\sqrt{-\dot\Sigma(t)}\dif t}.
\end{equation}
Note that $m(s)$ is deterministic, and so is not affected by the change of measure. The key feature of \eqref{7xqVkCTVPT} is that the RHS contains no random data, and so the problem \eqref{dSsYjBUGdd} may be solved using classical calculus of variations. In fact, writing \eqref{7xqVkCTVPT} in the form
\begin{equation*}
\max_{\dot{\Sigma}<0}\E{\int_0^T\sigma_t\sqrt{-\dot\Sigma(t)}\dif t} = \int_0^T\sigma_0\exp\!\del{\int_0^t m(s)\dif s}\sqrt{-\dot\Sigma(t)}\dif t,
\end{equation*}
we see that the RHS is of the form \eqref{BAomThdW31} with $\sigma_0\exp\!\del{\int_0^t m(s)\dif s}$ in place of $\sigma(t)$. The optimal $\dot\Sigma^*$, according to \eqref{GlDv15Ztsy}, is
\begin{equation}
\dot\Sigma^*(t) = - \Sigma_0 \frac{\exp\!\del{2\int_0^t m(s) \dif s}}{\int_0^T \exp\!\del{2\int_0^t m(s) \dif s} \dif t}.
\end{equation}
According to \eqref{BUSu1MNQ9m}, the optimal $M^*$ is
\begin{equation}
\label{2BhBl8nIM3}
M^*_t = \sigma_t\exp\!\del{-\int_0^t m(s) \dif s}\sqrt{\frac{1}{\Sigma_0}\int_0^T\exp\!\del{2\int_0^t m(s)\dif s}\dif t}.
\end{equation}

\subsubsection{The dual problem}

The dual problem \eqref{RUHFk3mIp0} reads
\begin{equation}
\label{hivcWPez1c}
\min_{M}\del{\frac{1}{2}\Sigma_0 M_0 + \frac{1}{2}\E{\int_0^T \frac{\sigma_t^2}{M_t}\dif t}}
\end{equation}
where the minimization is over all positive martingales $(M_t)_{0\le t\le T}$. If $M_t$ evolves according to
\begin{equation}
\label{dv5llapEmJ}
\dif M_t = \alpha_t M_t \dif W_t
\end{equation}
for some $\sF^{W}$ adapted $\alpha_t$, then due to \eqref{dpc20W9KFz}, the integrand in \eqref{hivcWPez1c} may be written
$$
\frac{\sigma^2_t}{M_t} = \frac{\sigma^2_0}{M_0}\exp\!\del{\int_0^t \del{2m(s)- \nu(s,\sigma_s)^2+\frac{\alpha_s^2}{2}}\dif s + \int_0^t\del{2\nu(s,\sigma_s)-\alpha_s}\dif W_s}.
$$
We introduce the measure $\dQ$ through the RN derivative
$$
\dod{\dQ}{\dP} = \exp\!\del{-\frac{1}{2}\int_0^t\del{2\nu(s,\sigma_s)-\alpha_s}^2 \dif s +\int_0^t\del{2\nu(s,\sigma_s)-\alpha_s}\dif W_s},
$$
and we compute
\begin{equation}
\label{XX7xhP6WSt}
\E{\frac{\sigma^2_t}{M_t}} = \dE_{\dQ}\!\sbr{\frac{\sigma^2_0}{M_0}\exp\!\del{\int_0^t 2m(s)- \del{\alpha_s - \nu(s,\sigma_s)}^2\dif s}}.
\end{equation}
The RHS integrand is minimized pointwise when
\begin{equation}
\label{zPuY9uxVIw}
\alpha_s = \nu(s,\sigma_s),
\end{equation}
in which case \eqref{XX7xhP6WSt} reduces to
$$
\E{\frac{\sigma^2_t}{M_t}} = \frac{\sigma^2_0}{M_0}\exp\!\del{2\int_0^t m(s)\dif s}.
$$
We substitute this into \eqref{hivcWPez1c} through an invocation of the tower property. When this is done, the minimization problem reduces to
$$
\min_{M_0}\del{\frac{1}{2}\Sigma_0 M_0 + \frac{1}{2}\frac{\sigma^2_0}{M_0}\int_0^T\exp\!\del{2\int_0^t m(s)\dif s}\dif t},
$$
and the optimal $M_0$ is given by
$$
M_0 = \sigma_0\sqrt{\frac{1}{\Sigma_0}\int_0^T\exp\!\del{2\int_0^t m(s)\dif s}\dif t}.
$$
In light of \eqref{dv5llapEmJ} and \eqref{zPuY9uxVIw}, we see now that the optimal martingale $(M_t)_{0\le t\le T}$ is
\begin{align*}
M_t &= \sigma_0\exp\!\del{-\int_0^t\frac{\nu(s,\sigma_s)^2}{2}\dif s + \int_0^t\nu(s,\sigma_s) \dif W_s}\sqrt{\frac{1}{\Sigma_0}\int_0^T\exp\!\del{2\int_0^t m(s)\dif s}\dif t} \\
 &= \sigma_t\exp\!\del{-\int_0^t m(s) \dif s}\sqrt{\frac{1}{\Sigma_0}\int_0^T\exp\!\del{2\int_0^t m(s)\dif s}\dif t}.
\end{align*}
By construction $(M_t)_{0\le t\le T}$ is a positive martingale; the assumed Novikov condition makes it a true martingale, and under the assumption that $\sigma$ is bounded above and below away from zero, \eqref{2BhBl8nIM3} shows that $M$ inherits these bounds and is therefore an interior minimizer. The terminal covariance identity holds by construction, so MDC holds, and Theorem~\ref{I3YrvmnwLc} asserts that $(M_t)_{0\le t\le T}$ represents the market depth in some equilibrium. The price impact of this same equilibrium is
\begin{equation}
\lambda_t = \frac{1}{M_t} = \frac{1}{\sigma_t} \exp\!\del{\int_0^t m(s) \dif s}\sqrt{\frac{\Sigma_0}{\int_0^T\exp\!\del{2\int_0^t m(s)\dif s}\dif t}},
\end{equation}
which is precisely (25) of \cite{cdf}.

\subsection{Scalar stochastic liquidity}
\label{sec:1D-verify}

For one-dimensional stochastic liquidity, the dual problem reduces to a scalar stochastic control problem, which we now solve via the stochastic maximum principle.

First, consider \eqref{RUHFk3mIp0} with $M_t$, $\sigma_t$, and $\Sigma_0$ all one dimensional. We split the optimization over $\del{M_t}_{0\leq t\leq T}$ into two steps: $\del{M_t}_{0<t\leq T}$ and $M_0$, in that order:
\begin{equation}
\label{oqehc4eMC6}
\inf_{M_0}\del{\frac{1}{2}\Sigma_0 M_0 + \frac{1}{2}\inf_{M_{(0,T]}}\E{\int_0^T \sigma^2_t M^{-1}_t\dif t}}
\end{equation}
The first step renders the objective value into a coercive function of $M_0$.

We will assign $M_t$ and $\sigma_t$ the dynamics
\begin{align}
\label{YNDdAtxo0M} \dif M_t &= \alpha_t M_t \dif W_t \\
\label{wZMHa6ElQE} \dif \sigma_t &= \sigma_t m_t \dif t + \sigma_t \nu_t \dif W_t
\end{align}
with $M_0$ and $\sigma_0$ fixed. Here, $m_t$ and $\nu_t$ are exogenously given processes which obey some regularity conditions. Introducing the state variable $S_t = \sigma^2_t M_t^{-1}$, the optimization over $(M_t)_{0<t\leq T}$ reformulates as the control problem
\begin{equation}
\label{ywfcmw769r}
\inf_{\alpha} \; \E{\int_0^T S_t \dif t}.
\end{equation}
$S_t$ is a controlled state variable with dynamics
\begin{equation}
\label{HuVBZsy9IR}
\dif S_t = \del{2 m_t + \del{\alpha_t - \nu_t}^2}S_t \dif t + (2\nu_t - \alpha_t)S_t \dif W_t
\end{equation}
and $S_0 = \sigma^2_0 M_0^{-1}> 0$ is fixed.

The first adjoint equation \cite[equation 3.8]{yongzhou} reads
\begin{align}
\label{g4DMsIzZOL}
\dif G_t &= -\del{\del{2m_t + \del{\alpha^*_t - \nu_t}^2}G_t + \del{2\nu_t - \alpha^*_t}U_t - 1}\dif t + U_t \dif W_t \\
G_T &= 0
\end{align}
and the second \cite[equation 3.9]{yongzhou}
\begin{align}
\label{mPcXX0DtKj}
\dif g_t &= -\del{\del{4m_t + 6\nu_t^2 + 3(\alpha^*_t)^2 - 8\alpha^*_t\nu_t}g_t + \del{4\nu_t - 2\alpha^*_t}u_t}\dif t + u_t \dif W_t \\
\label{6Ev69iHvlK} g_T &= 0.
\end{align}
The solution to \eqref{mPcXX0DtKj}--\eqref{6Ev69iHvlK} is $\del{g_t, u_t} = (0,0)$. The $\cH$ function \cite[equation 3.16]{yongzhou} therefore coincides with the Hamiltonian  \cite[equation 3.10]{yongzhou}
\begin{equation}
\label{RBOLLfNs9K}
H(t,s,\alpha,G,U) = \del{2m_t + \nu_t^2 + \alpha^2 - 2\alpha\nu_t}sG + \del{2\nu_t - \alpha}sU - s.
\end{equation}
The stochastic maximum principle \cite[theorem 3.2]{yongzhou} asserts that the optimal $\bar{\alpha}$ is a pointwise minimizer of $H\del{t,\bar{X}_t, \alpha, G_t, U_t}$. We find
\begin{equation}
\label{XXZtsnLrMz}
\alpha^*_t = \nu_t + \frac{1}{2}\frac{U_t}{G_t}.
\end{equation}
Substituting this into \eqref{g4DMsIzZOL}, we obtain
\begin{align}
\label{3JfVRAFuX7} \dif G_t &= \del{\frac{U_t^2}{4G_t} - \nu_t U_t - 2m_t G_t +1}\dif t + U_t \dif W_t \\
\label{NyunRhmZ9w} G_T &= 0
\end{align}
If we define the probability measure $\dP_{\nu}$ through the Radon-Nikodym derivative
$$
\frac{\dif \dP_{\nu}}{\dif \dP} = \exp\!\del{-\frac{1}{2}\int_0^T \nu_t^2 \dif t + \int_0^T \nu_t \dif W_t}
$$
then we may write \eqref{3JfVRAFuX7} as
\begin{equation}
\label{jsYH8ONBHw}
\dif G_t = \del{\frac{U_t^2}{4G_t}-2m_t G_t + 1}\dif t + U_t \dif W^{\nu}_t
\end{equation}
where $W^{\nu}_t = W_t - \int_0^t\nu_s \dif s$ is a $\dP_{\nu}$ Brownian motion by Girsanov's theorem. Furthermore, we may multiply both sides of the above by the integrating factor $\exp\!\del{\int_0^t m_s \dif s}$ and change variables to
\begin{align}
\tilde G_t&= -e^{2\int_0^t m_s ds}G_t \\
\tilde U_t&= -e^{2\int_0^t m_s \dif s}U_t.
\end{align}
Then, the BSDE \eqref{NyunRhmZ9w}, \eqref{jsYH8ONBHw} becomes
\begin{align}
\dif \tilde G_t &= \del{\frac{\tilde U_t^2}{4\tilde G_t}  -e^{2\int_0^t m_s \dif s}}\dif t +\tilde U_t \dif W^{\nu}_t\\
\tilde G_T &= 0,
\end{align}
which is the same as that appearing in \cite[proposition 2.1]{emz}. Invoking this result, we obtain the pair $\del{G_t, U_t}$. The control \eqref{XXZtsnLrMz} is well defined, and we can substitute it into \eqref{HuVBZsy9IR} to obtain the optimal state trajectory 
\begin{equation}
\label{QHucANH16T}
\bar{S}_t = \sigma_0^2 M_{0}^{-1} \exp\!\del{\int_0^t \del{2m_s - \frac{\nu^2_s}{2} + \nu_s\frac{U_s}{2G_s} + \frac{U_s^2}{8G^2_s}} \dif s + \int_0^t \del{\nu_s - \frac{U_s}{2G_s}}\dif W_s}.
\end{equation}
If we substitute \eqref{QHucANH16T} into \eqref{ywfcmw769r}, we will find that the optimal value is of the form $C \sigma_0^2 M_{0}^{-1}$ for some $C>0$. Thus the original objective \eqref{oqehc4eMC6} is coercive in $M_0$, and admits a minimizer. Altogether, the optimal $M_t$ is
\begin{equation}
\bar{M}_t = {\sigma_t^2}/{\bar{S}_t} =  M_0\exp\!\del{- \frac{1}{2}\int_0^t \del{\nu_s + \frac{U_s}{2G_s}}^2 \dif s + \int_0^t \del{\nu_s + \frac{U_s}{2G_s}} \dif W_s}.
\end{equation}

Under the structural assumptions of \cite{emz}, the integrand $\nu_s+U_s/(2G_s)$ lies in $\mathrm{bmo}$, so $\bar M$ is a true $\sF^{\bW}$-martingale. If $\bar M$ is in addition uniformly bounded away from zero---for instance when $\sigma$ is bounded above and below away from zero, as in \cite{cdf}---then $\bar M$ is an interior minimizer, whence Proposition~\ref{prop:MDC-terminal-covariance} gives the terminal covariance identity and MDC is satisfied. In particular, the scalar theory subsumes the models of \cite{kyle, back1992, bp, cdf, emz}.

\subsection{Common eigenbasis}
\label{sec:eigenbasis-verify}
The above result extends to a special case of the multidimensional problem \eqref{RUHFk3mIp0}.

\begin{proposition}
Suppose  $\bC$ and $(\bfsigma_t)$ are diagonalizable along the same basis for all $0 \leq t \leq T$:
\begin{equation}
\label{PTYXvuAbrQ}
\bC = 
\bV
\begin{bmatrix}
\Sigma_0^1 & & \\
& \ddots & \\
& & \Sigma_0^n
\end{bmatrix}
\bV^{\top},
\qquad
\bfsigma_t = \bV
\begin{bmatrix}
\sigma_t^1 & & \\
& \ddots & \\
& & \sigma_t^n
\end{bmatrix}
\bV^{\top}.
\end{equation}
Here, $\bV$ is a fixed orthogonal matrix and, in accordance with \eqref{wZMHa6ElQE}, the diagonal entries $\sigma_t^i$ solve
$$
\dif \sigma_t^i = \sigma_t^i m_t^i \dif t + \sigma_t^i \nu_t^i \dif W_t^i, \qquad i = 1,\dots,n.
$$
With this assumed form on $\bfsigma_t$, the minimizer of \eqref{RUHFk3mIp0} is
\begin{equation}
\label{9mRvQhaTPa}
\bar{\bM}_t = \bV
\begin{bmatrix}
\bar{M}^1_t & & \\
& \ddots & \\
& & \bar{M}^n_t
\end{bmatrix}
\bV^{\top},
\end{equation}
where each $\bar{M}^i_t$ minimizes the one-dimensional problem
$$
\inf_{M^i}\del{\frac{1}{2}\Sigma_0^i M^i_0 + \frac{1}{2}\E{\int_0^T \del{\sigma^i_t}^2 \del{M^i_t}^{-1}\dif t}}.
$$
\end{proposition}

The proof is deferred to appendix~\ref{app:verification-proofs}.

\subsection{The general matrix-valued case}
\label{sec:open-problem}

When $\bC$ and $(\bfsigma_t)$ do not share a time-independent eigenbasis, neither of the verification arguments above applies, and we leave open the question of whether MDC holds. We record here the natural fully-coupled FBSDE that arises from applying the stochastic maximum principle (in the spirit of section~\ref{sec:1D-verify}) to the dual problem in this general setting.

Parametrize $\bM_t$ by $\dif \bM_t = \bM_t \tilde{\bfalpha}_t \bM_t\,\dif\bW_t$, so that $\bfLambda_t = \bM_t^{-1}$ satisfies
\[
\dif \bfLambda_t = \tilde{\bfalpha}_t \bfLambda_t^{-1} \tilde{\bfalpha}_t\,\dif t - \tilde{\bfalpha}_t\,\dif\bW_t.
\]
The stochastic maximum principle then yields an adjoint $\tilde{\bG}_t$ with companion process $\tilde{\bU}_t$ satisfying, formally,
\begin{align}
\label{eq:open-FBSDE-G}
\dif \tilde{\bG}_t &= \bM_t \tilde{\bfalpha}^*_t \tilde{\bG}_t \tilde{\bfalpha}^*_t \bM_t\,\dif t + \bfsigma_t^2\,\dif t + \tilde{\bU}_t\,\dif\bW_t, \\
\label{eq:open-FBSDE-M}
\dif \bM_t &= \bM_t \tilde{\bfalpha}^*_t \bM_t\,\dif\bW_t,
\end{align}
where $\tilde{\bfalpha}^*$ is determined pointwise by the Lyapunov equation
\[
\tilde{\bG}\tilde{\bfalpha}^*_t \bM_t + \bM_t \tilde{\bfalpha}^*_t \tilde{\bG} = \tilde{\bU}.
\]
Vectorizing, \eqref{eq:open-FBSDE-G}--\eqref{eq:open-FBSDE-M} is an $n^2$-dimensional fully coupled FBSDE in which (i) the diffusion coefficient of the forward equation depends on the BSDE martingale integrand and (ii) the driver is quadratic in the integrand and singular in the adjoint state. Wellposedness for fully coupled FBSDEs of this type is, to the best of our knowledge, an open problem; the quadratic BSDE techniques of \cite{emz} address (ii) but not (i). We leave both the wellposedness of \eqref{eq:open-FBSDE-G}--\eqref{eq:open-FBSDE-M} and the verification of MDC in the general matrix-valued setting to future~work.

\section{Economic interpretation}
\label{sec:economic-interpretation}

The variational problem treats private information as an inventory. The state variable $\bfSigma$ is the market maker's remaining posterior covariance, and the control $-\dot\bfSigma$ is the instantaneous rate at which the insider releases information into prices. Thus the insider does not merely liquidate shares; she liquidates informational advantage.

The dual martingale $\bM^*$ is the shadow price of this information inventory. Its inverse,
\[
\bfLambda^*=(\bM^*)^{-1},
\]
is the matrix-valued price impact. In one dimension, $M^*$ is exactly market depth. In the stochastic-liquidity model of Collin-Dufresne and Fos, this identifies their auxiliary process $G$ as the costate process associated with the dual problem.

Stochastic liquidity matters because it changes the optimal timing of information release. When noise trading is high, informed trades are better camouflaged, so the optimal release rate increases. In equilibrium, market makers anticipate this timing option, and the martingale condition on market depth compensates the insider for giving up the option to wait for better liquidity states.

\section{Conclusion}
In this paper, we developed a variational reformulation of the Kyle--Back model with stochastic and multidimensional liquidity. By recasting equilibrium as a primal--dual problem over decreasing covariance paths, we identified informed trading with an \emph{optimal liquidation of information}: the insider's control variable is not her share inventory but the rate at which private information is released into the price. This reformulation unifies, within a single convex-analytic and optimal-transport framework, the equilibria of Kyle
\cite{kyle}, Back \cite{back1992}, Back--Pedersen \cite{bp}, Collin-Dufresne--Fos \cite{cdf}, and the multi-asset model of Back--Cocquemas--Ekren--Lioui \cite{cel2020}, each of which we recover by specialization.

The dual problem, formulated in terms of a matrix-valued martingale $\bM$, has two payoffs. First, it gives a clean economic interpretation of the auxiliary process $G$ of \cite{cdf}: it is the costate variable for the dual control problem and so plays the role of a \emph{shadow price of information}. Second, it isolates a single structural hypothesis on the data---the martingale dual condition MDC of Definition~\ref{def:MDC}---under which equilibrium exists with pricing rule $\bp_0 + \bfxi$ and price impact $\bfLambda = (\bM^{*})^{-1}$. We verified MDC directly in the scalar setting and in the multidimensional setting with a common time-independent eigenbasis between $\bC$ and~$\bfsigma$.
Whether MDC holds for general $\bC$ and $\bfsigma$ remains open, and reduces, via the stochastic maximum principle, to the wellposedness of a fully coupled FBSDE with quadratic and singular driver (section~\ref{sec:open-problem}).

Along the way, we also established an independently interesting Doob--Meyer decomposition theorem for general matrix-valued local submartingales (Theorem~\ref{YJoB18UgVC}).

\appendix

\section{Heuristic derivation from causal transport}\label{tBdCdebe09}

Fix \(t\in[0,T]\) and a current covariance matrix \(\bfSigma\in\cS^n_+\). Let
\(\vt\sim\cN({\bf 0},\bfSigma)\). We start with the optimal transport type problem
\begin{equation}
\label{cot}
V(t,\bfSigma)
=
\sup_{\pi}
\Epi{\vt^{\top} \int_t^T \bfsigma_s \dif \bB_s \middle| \sF_t}
\end{equation}
where the supremum is taken over all couplings \(\pi\) of \((\vt,\bB,\bW)\) such that
\(\vt\) has law \(\cN({\bf 0},\bfSigma)\), \(\vt\) and \(\bW\) are independent, and \(\bB\)
is Brownian in its own filtration and is coupled with \((\vt,\bW)\) causally. In the model
of the main text, the initial covariance is \(\bC\), so the corresponding initial value is
\(V(0,\bC)\).

Introducing the partition $t = t_0 < \cdots < t_n = T$, we approximate \eqref{cot} as
$$
V^n(t_k,\bfSigma) = \sup_{\pi_k}\sum_{k=0}^{n-1}\Epi{\vt^{\top} \bfsigma_{t_k}\del{\bB_{t_{k+1}}-\bB_{t_k}}}
$$
or, written dynamically, as
\begin{equation}
\label{dynamic problem}
V^n(t_k,\bfSigma) = \sup_{\pi_k}\Epi{\vt^{\top}\bfsigma_{t_k}\del{\bB_{t_{k+1}}-\bB_{t_k}} + V^n\!\del{t_{k+1},\bfSigma_{t_{k+1}}}\,\middle|\,\sF_t}
\end{equation}
where
\begin{align}
\label{filtered covariance}
\bfSigma_{t_{k+1}} = \var{\vt \middle| \sF^{\bB}_{t_{k+1}}}.
\end{align}
To couple $\del{\vt, \bB_{t_{k+1}}-\bB_{t_k}}$, let the pair be jointly Gaussian with covariance matrix
$$
\var{
\begin{bmatrix}
\vt \\
\bB_{t_{k+1}}-\bB_{t_k}
\end{bmatrix}
\middle| \sF^{\bB}_{t_{k}}}
=
\begin{bmatrix}
\bfSigma & \sqrt{\bfSigmat_k}\bQ_k\Delta t_k \\
\bQ_k^{\top}\sqrt{\bfSigmat_k}\Delta t_k & \bI \Delta t_k
\end{bmatrix}.
$$
Note that we have polar factorized
\begin{equation}
\label{cov v dB}
\cov{\vt}{\bB_{t_{k+1}}-\bB_{t_k}} = \sqrt{\bfSigmat_k}\bQ_k \Delta t_k
\end{equation}
so that $\bfSigmat_k$ is symmetric positive semi-definite and $\bQ_k$ is orthogonal. Invoking the Gaussian projection theorem, \eqref{cov v dB} and \eqref{filtered covariance} yield
$$
\bfSigma_{t_{k+1}} = \bfSigma - \bfSigmat_k \Delta t_k.
$$
Similarly, \eqref{cov v dB} permits the decomposition
$$
\vt = \bz + \sqrt{\bfSigmat_k}\bQ_k \del{\bB_{t_{k+1}}-\bB_{t_k}}
$$
where $\bz$ is independent of $\bB_{t_{k+1}}-\bB_{t_k}$. With these considerations, \eqref{dynamic problem} updates to
\begin{align*}
V^n&(t_k,\bfSigma)  \\
& = \sup_{\bfSigmat_k}\sup_{\bQ_k}\Epi{\del{\bB_{t_{k+1}}-\bB_{t_k}}^{\top}\bQ_k^{\top}\sqrt{\bfSigmat_k}\bfsigma_{t_k}\del{\bB_{t_{k+1}}-\bB_{t_k}} + V^n\!\del{t_{k+1},\bfSigma - \bfSigmat_k \Delta t_k}} \\
&=\sup_{\bfSigmat_k}\sup_{\bQ_k}\Epi{\tr{\bQ_k^{\top}\sqrt{\bfSigmat_k}\bfsigma_{t_k}}\Delta t_k + V^n\!\del{t_{k+1},\bfSigma - \bfSigmat_k \Delta t_k}}.
\end{align*}
To attain the supremum over $\bQ_k$, write $\sqrt{\bfSigmat_k}\bfsigma_{t_k} = \bO_k \bP_k$ for $\bO_k$ orthogonal and $\bP_k$ positive, and choose $\bQ_k = \bO_k$. The trace term reduces to
$$
\tr{\bP_k} = \tr{\sqrt{\bfsigma_{t_k}^{\top}\bfSigmat_k\bfsigma_{t_k}}},
$$
and $V^n(t_k,\bfSigma)$ to
$$
V^n(t_k,\bfSigma) = \sup_{\bfSigmat_k}\Epi{\tr{\sqrt{\bfsigma_{t_k}^{\top}\bfSigmat_k\bfsigma_{t_k}}}\Delta t_k + V^n(t_{k+1},\bfSigma - \bfSigmat_k \Delta t_k)}
$$
where $\bfSigmat_k$ is chosen so that $\bfSigma - \bfSigmat_k \Delta t_k$ remains positive semi-definite. This motivates our continuous-time problem
\begin{align}
\label{continuous time problem}
V(t,\bfSigma) &= \sup_{\bfSigmat}\cbr{\E{\int_t^T \tr{\sqrt{\bfsigma_s \bfSigmat_s \bfsigma_s}}\dif s \middle| \sF^{\bW}_t} \colon \int_t^T \bfSigmat_s \dif s \leq \bfSigma}\\
&= \sup_{(\bfSigmat_s,\bQ_s)\in \dR^{n\times n}\times O_n}\cbr{\E{\int_t^T \tr{\bfSigmat_s\bfsigma_s \bQ_s}\dif s \middle| \sF^{\bW}_t} \colon \int_t^T \bfSigmat_s^\top\bfSigmat_s \dif s \leq \bfSigma}
\end{align}
or, written dynamically,
$$
V(t,\bfSigma)\hspace{-0.2mm} =\hspace{-0.2mm} \sup_{\bfSigmat}\cbr{\hspace{-0.2mm}\E{\int_t^{\tau}\hspace{-0.2mm} \tr{\sqrt{\bfsigma_s \bfSigmat_s \bfsigma_s}}\dif s\hspace{-0.2mm} +\hspace{-0.2mm} V\!\del{\tau, \bfSigma - \int_{t}^{\tau}\bfSigmat_s \dif s}\middle| \sF^{\bW}_t} \colon \int_t^{\tau}\hspace{-0.2mm}\bfSigmat_s \dif s \leq \bfSigma}.
$$

\section{Matrix-valued Doob--Meyer decomposition theorem}
\label{app:matrix-doob-meyer}

The purpose of this section is to establish a matrix-valued generalization of the classical (scalar-valued) Doob--Meyer decomposition theorem, the statement of which is recalled below.
Throughout the section, $(\Omega,\sF,(\sF_t)_{t \geq 0},\dP)$ is a filtered probability space satisfying the usual conditions, and all processes are assumed to be c\`adl\`ag and adapted to $(\sF_t)_{t \geq 0}$ unless otherwise stated.

\begin{theorem}[Scalar-valued Doob--Meyer decomposition theorem]\label{thm.scalarDMdecomp}
If $X = (X_t)_{t \geq 0}$ is a local submartingale, then there exist unique-up-to-indistinguishability processes $M = (M_t)_{t \geq 0}$ and $A = (A_t)_{t \geq 0}$ such that $M$ is a local martingale, $A$ is a non-decreasing predictable process, $A_0 = 0$ almost surely, and $X=M+A$.
Furthermore:
\begin{enumerate}[label=(\roman*),font=\normalfont]
    \item If $X$ is continuous, then so are $A$ and $M$;\label{item.cont}
    \item $X$ is of class (DL) if and only if $M$ is a true martingale and $A$ is integrable;\label{item.DL}
    \item $X$ is of class (D) if and only if $M$ is a uniformly integrable martingale and $A_{\infty} \coloneqq \lim_{t \to \infty} A_t$ is integrable.\label{item.D}
\end{enumerate}
\end{theorem}

To obtain the definitions of matrix-valued (local) submartingales, (local) supermartingales, non-decreasing processes, and non-increasing processes, simply write down the scalar-valued definitions, interpret the inequalities as inequalities in the Loewner order, and replace $|\cdot|$ with $\|\cdot\|_2$.
For example, here is the definition of a matrix-valued supermartingale.

\begin{definition}
\label{sySq817Omx}
Let $\bX = (\bX_t)_{t \geq 0}$ be an adapted $n \times n$ matrix--valued process.
We say that $\bX$ is a {\it supermartingale} if $\E{\lVert\bX_t\rVert_2} < \infty$ for each $t \geq 0$ and
$$
\E{\bX_t \middle| \sF_s} \leq \bX_s, \qquad 0 \leq s \leq t,
$$
where the above inequality means $\bX_s - \E{\bX_t \middle| \sF_s}$ is a positive semi-definite matrix.
\end{definition}

Similarly, to define uniform integrability for collections of matrix-valued random variables, simply replace $|\cdot|$ with $\|\cdot\|_2$ in the scalar-valued definition.
This yields, in particular, obvious notions of matrix-valued class (D) and class (DL) processes.

We are now prepared for the main result of this section.

\begin{theorem}[Matrix-valued Doob--Meyer decomposition theorem]\label{YJoB18UgVC}
If $\bX$ is an $n \times n$ matrix--valued local submartingale, then there exist unique-up-to-indistinguishability $n\times n$ matrix--valued processes $\bM = (\bM_t)_{t \geq 0}$ and $\bA = (\bA_t)_{t \geq 0}$ such that $\bM$ is a local martingale, $\bA$ is a non-decreasing predictable process, $\bA_0 = \bf0$ almost surely, and $\bX=\bM+\bA$.
Furthermore:
\begin{enumerate}[label=(\roman*),font=\normalfont]
    \item If $\bX$ is continuous, then so are $\bA$ and $\bM$;\label{item.matcont}
    \item $\bX$ is of class (DL) if and only if $\bM$ is a true martingale and $\bA$ is integrable;\label{item.matDL}
    \item $\bX$ is of class (D) if and only if $\bM$ is a uniformly integrable martingale and $\bA_{\infty} \coloneqq \lim_{t \to \infty} \bA_t$ is integrable.\label{item.matD}
\end{enumerate}
\end{theorem}

\begin{proof}
Let us begin by establishing the uniqueness statement.
To this end, suppose $\bX = \bA^1+\bM^1=\bA^2+\bM^2$, where $\bM^1$ and $\bM^2$ are local martingales and $\bA^1$ and $\bA^2$ are non-decreasing predictable processes starting at zero.
Since $\bA^1$ and $\bA^2$ are non-decreasing and thus finite-variation processes, the process $\bM^1-\bM^2 = (\bX - \bA^1) - (\bX - \bA^2) = \bA^2-\bA^1$ is a predictable, finite-variation local martingale starting at zero.
Therefore, $\bM^1-\bM^2 = \bA^2-\bA^1$ is indistinguishable from the constant $\bf 0$.

To establish the decomposition's existence, we first argue that $\bY \coloneqq \big(\bX-\bX^{\top}\big)/2$ is a local martingale.
By localization, it suffices to argue that if $\bX$ is a submartingale, then $\bY$ is a martingale.
In this case, suppose $0 \leq s \leq t$.
Since $\bX$ is a submartingale, $\E{\bX_t \middle| \sF_s} - \bX_s \geq \bf0$.
In particular, $\E{\bX_t \middle| \sF_s} - \bX_s$ is a symmetric matrix, i.e.,
\[
\E{\bX_t \middle| \sF_s} - \bX_s = \left(\E{\bX_t \middle| \sF_s} - \bX_s\right)^{\top} = \E{\bX_t^{\top} \middle| \sF_s} - \bX_s^{\top}.
\]
Rearranging yields
\[
\E{\bY_t \middle| \sF_s} = \frac12 \left(\E{\bX_t \middle| \sF_s} - \E{\bX_t^{\top} \middle| \sF_s}\right) = \frac12\left(\bX_s - \bX_s^{\top}\right) = \bY_s,
\]
as desired.
\pagebreak

Now, suppose Doob--Meyer decompositions exist for symmetric matrix--valued local submartingales.
Since $\bZ \coloneqq \big(\bX+\bX^{\top}\big)/2$ is a symmetric matrix--valued local submartingale, there exist a local martingale $\bN$ and a non-decreasing predictable process $\bA$ starting at zero such that $\bZ = \bN + \bA$.
By the previous paragraph, $\bM \coloneqq \bY+\bN = \big(\bX-\bX^{\top}\big)/2+\bN$ is a local martingale.
Since $\bX = \bY+\bZ = \bY+\bN+\bA = \bM+\bA$, we conclude that $\bX$ has a Doob--Meyer decomposition.
It therefore suffices to treat the case in which $\bX$ takes values in the symmetric matrices.

Suppose $\bX$ takes values in the symmetric matrices.
Also, for an $n \times n$ matrix $\bB$, write
\[
\bB^{\bu,\bv} \coloneqq \bu^{\top}\bB\bv \quad \text{and} \quad \bB^{\bu} \coloneqq \bB^{\bu,\bu}, \qquad \bu,\bv \in \dR^n.
\]
Let $\bu \in \dR^n$.
By definition of matrix-valued local submartingale and the Loewner order, $\bX^{\bu}$ is a real-valued local submartingale.
By the scalar-valued Doob--Meyer decomposition theorem, there exist a local martingale $M^{(\bu)}$ and a non-decreasing predictable process $A^{(\bu)}$ starting at zero such that $\bX^{\bu} = M^{(\bu)} + A^{(\bu)}$.
Observe that since $\bX^{\bu} = \bX^{-\bu}$, the uniqueness part of the scalar-valued Doob--Meyer decomposition theorem yields that $M^{(\bu)}$ (respectively, $A^{(\bu)}$) is indistinguishable from $M^{(-\bu)}$ (respectively, $A^{(-\bu)}$).

Define
\[
\bM_{ij} \coloneqq \frac14\left( M^{(\be_i+\be_j)} - M^{(\be_i-\be_j)}  \right) \; \text{ and } \; \bA_{ij} \coloneqq \frac14\left( A^{(\be_i+\be_j)} - A^{(\be_i-\be_j)}  \right), \quad i,j=1,\ldots,n,
\]
where $\be_1,\ldots,\be_n$ is the standard basis of $\dR^n$.
Since the entries of $\bM$ are local martingales, $\bM$ is a matrix-valued local martingale.
Since the entries of $\bA$ are predictable finite-variation processes starting at zero, $\bA$ is a matrix-valued finite-variation predictable process starting at zero.
Furthermore, since $\bX$ takes values in the symmetric matrices,
\begin{align*}
    \bX_{ij} & = \bX^{\be_i,\be_j} = \frac14\left(\bX^{\be_i+\be_j} - \bX^{\be_i-\be_j}\right) \\
    & = \frac14\left( M^{(\be_i+\be_j)} - M^{(\be_i-\be_j)}\right) + \frac14\left(A^{(\be_i+\be_j)} - A^{(\be_i-\be_j)}\right) \\
    & = \bM_{ij} + \bA_{ij}
\end{align*}
for all $i,j=1,\ldots,n$.
Thus, $\bX = \bM+\bA$.
It remains to show that $\bA$ is a non-decreasing process.
To this end, observe that since $A^{(\bu)}$ is indistinguishable from $A^{(-\bu)}$ for all $\bu \in \dR^n$, $\bA_{ij}$ is indistinguishable from $\bA_{ji}$ for all $i,j=1,\ldots,n$.
We therefore may and do assume $\bA$ takes values in the symmetric matrices.
Now, fix $\bu \in \dR^n$.
Since $\bM^{\bu}$ is a local martingale, $\bA^{\bu}$ is a predictable finite-variation process starting at zero, and $\bX^{\bu} = \bM^{\bu} + \bA^{\bu}$, a slightly strengthened version of the uniqueness part of the scalar-valued Doob--Meyer decomposition theorem---see the argument in the first paragraph of this proof---yields that $\bA^{\bu}$ is a non-decreasing process.
In other words, there exists a full-probability set $\Om_{\bu} \in \sF_0$ such that on $\Om_{\bu}$, $\bA_s^{\bu} \leq \bA_t^{\bu}$ for all $t \geq s \geq 0$.
Now, let $\Om_0 \coloneqq \bigcap_{\bu \in \dQ^n} \Om_{\bu}$.
On the full-probability set $\Om_0$,
\[
\bu^{\top}\bA_s \bu \leq \bu^{\top}\bA_t\bu, \qquad 0 \leq s \leq t, \; \bu \in \dQ^n.
\]
It then follows from the continuity of $\bu \mapsto \bB^{\bu}$ for an $n \times n$ matrix $\bB$ that the above inequality holds for \emph{all} $\bu \in \dR^n$ and $t \geq s \geq 0$ on $\Om_0$.
Thus, $\bA$ is a non-decreasing process.
This completes the proof of the existence of matrix-valued Doob--Meyer decompositions.

By studying the entries of $\bX$, $\bM$, and $\bA$, items~\ref{item.matcont}--\ref{item.matD} follow readily from the corresponding items in Theorem~\ref{thm.scalarDMdecomp}.
\end{proof}

\begin{remark}\label{rem.supermart}
By considering $-\bX$ in place of $\bX$ in Theorem~\ref{YJoB18UgVC}, we obtain that if $\bX$ is an $n \times n$ matrix--valued local supermartingale, then there exist unique-up-to-indistinguishability $n \times n$ matrix--valued processes $\bL$ and $\bA$ such that $\bL$ is a local martingale, $\bA$ is a non-increasing predictable process starting at zero, and $\bX = \bL + \bA$.
Furthermore, Theorem~\ref{YJoB18UgVC} has an analog for  matrix-valued local submartingales $\bY = (\bY_t)_{0 \leq t \leq T}$ defined only on bounded time intervals;
simply apply Theorem~\ref{YJoB18UgVC} to the process $\bX = (\bX_t)_{t \geq 0}$, where $\bX_t \coloneqq \bY_t$ for all $t \in [0,T]$ and $\bX_t \coloneqq \bY_T$ for all $t > T$.
\end{remark}

\section{Proofs for the variational duality}
\label{app:dual-proofs}

\begin{proof}[Proof of Proposition~\ref{prop:MDC-terminal-covariance}]
Let $\bN$ be a bounded symmetric $\sF^{\bW}$-martingale. By the interior
property, $\bM^*+\varepsilon\bN\in\cM$ for all sufficiently small
$|\varepsilon|$. Since $\bM^*$ minimizes \eqref{RUHFk3mIp0}, we have
\[
    \dod{}{\varepsilon}
    g(\bM^*+\varepsilon\bN)\bigg|_{\varepsilon=0}=0,
\]
where
\[
g(\bM)
=
\frac12\tr{\bC\bM_0}
+
\frac12\E{\int_0^T \tr{\bfsigma_t^2\bM_t^{-1}}\,\dif t}.
\]
Using
\[
    \dod{}{\varepsilon}
    (\bM_t^*+\varepsilon\bN_t)^{-1}\bigg|_{\varepsilon=0}
    =
    -(\bM_t^*)^{-1}\bN_t(\bM_t^*)^{-1},
\]
we obtain
\[
0
=
\frac12\tr{\bC\bN_0}
-
\frac12
\E{\int_0^T
\tr{
(\bM_t^*)^{-1}\bfsigma_t^2(\bM_t^*)^{-1}\bN_t}\,\dif t}.
\]
Set
\[
    \bA_t^*
    :=
    (\bM_t^*)^{-1}\bfsigma_t^2(\bM_t^*)^{-1}.
\]
Since $\bN$ is a true martingale,
\[
    \bN_t=\E{\bN_T\mid\sF_t^{\bW}},
\]
and hence
\[
    \E{\int_0^T \tr{\bA_t^*\bN_t}\,\dif t}
    =
    \E{\tr{
        \left(\int_0^T\bA_t^*\,\dif t\right)\bN_T
    }}.
\]
Also,
\[
    \tr{\bC\bN_0}=\E{\tr{\bC\bN_T}}.
\]
Therefore
\[
    \E{\tr{
        \left(
            \bC-\int_0^T\bA_t^*\,\dif t
        \right)\bN_T
    }}=0\pagebreak
\]
for every bounded symmetric terminal perturbation $\bN_T$. It follows that
\[
    \bC=\int_0^T\bA_t^*\,\dif t
    =
    \int_0^T
    (\bM_t^*)^{-1}\bfsigma_t^2(\bM_t^*)^{-1}\,\dif t.
\]
The representation of $\bfSigma^*$ and the identities
$\bfSigma_0^*=\bC$ and $\bfSigma_T^*={\bf 0}$ follow immediately.
\end{proof}

\begin{proof}[Proof of Proposition~\ref{wn0NtKxptx}]
Let $\bM_{[0,T]} \in \cM$ be given. Choose
\begin{equation}
\label{DzvEJ7XTQT}
\dot{\bfSigma}^*_{t} = -\bM_t^{-1} \bfsigma^2_t \bM_t^{-1}, \qquad t \in [0,T]
\end{equation}
and note that $\dot{\bfSigma}^*_{[0,T]}\in\cA$ because
$\dot{\bfSigma}^*_t\in\cS^n_{-}$ for each $t$ and
$\dot{\bfSigma}^*_t$ is adapted to $\sF_t^{\bW}$.
In fact, $\dot{\bfSigma}^*_t\in\cS^n_{--}$ pointwise, since
\[
-\dot{\bfSigma}^*_t
=
\bM_t^{-1}\bfsigma_t^2\bM_t^{-1}
\in\cS^n_{++}.
\]
Thus the pointwise maximizer of the Lagrangian integrand lies in the interior
of the cone $\cS^n_{-}$, so the first-order condition below carries no boundary
normal term.
At this stage we do not claim that $\dot{\bfSigma}^*\in\cA(\bC)$; the terminal
constraint is recovered only at a dual minimizer satisfying
\[
\bC+\int_0^T\dot{\bfSigma}^*_t\,\dif t=\bf 0.
\]

Consider an admissible perturbation $\dot{\bfSigma}^* + \varepsilon \dot{\bfeta}$. We will show that \eqref{DzvEJ7XTQT} satisfies the first order condition for the optimality of 
$$\cA \ni \dot{\bfSigma} \mapsto L(\dot{\bfSigma}_{[0,T]},\bM_{[0,T]}),$$
which reads
\begin{equation}
\label{foXZfSGcUD}
\lim_{\varepsilon \to 0}\frac{1}{\varepsilon}\E{\int_0^T\del{\tr{\sqrt{-\bfsigma_t \del{\dot{\bfSigma}^*_t+\varepsilon \dot{\bfeta}_t} \bfsigma_t} - \sqrt{-\bfsigma_t \dot{\bfSigma}^*_t \bfsigma_t} + \frac{1}{2} \varepsilon \bM_t \dot{\bfeta}_t}}\dif t} = 0.
\end{equation}
From \cite{moralniclas2018}, we have the first order Taylor approximation
$$
\sqrt{-\bfsigma_t \del{\dot{\bfSigma}^*_t+\varepsilon \dot{\bfeta}_t} \bfsigma_t} = \sqrt{-\bfsigma_t \dot{\bfSigma}^*_t \bfsigma_t} - \varepsilon\int_0^{\infty}e^{-s\sqrt{-\bfsigma_t \dot{\bfSigma}^*_t \bfsigma_t}}\bfsigma_t\dot{\bfeta}_t\bfsigma_t e^{-s\sqrt{-\bfsigma_t \dot{\bfSigma}^*_t \bfsigma_t}} \dif s + O(\varepsilon^2).
$$
With this, the expected value
in \eqref{foXZfSGcUD} reduces to
$$
\E{\int_0^T\tr{\frac{1}{2} \bM_t \dot{\bfeta}_t - \int_0^{\infty}e^{-s\sqrt{-\bfsigma_t \dot{\bfSigma}^*_t \bfsigma_t}}\bfsigma_t\dot{\bfeta}_t\bfsigma_t e^{-s\sqrt{-\bfsigma_t \dot{\bfSigma}^*_t \bfsigma_t}} \dif s}\dif t}.
$$
We cyclically permute the integrand within the trace to obtain
\begin{align*}
&\E{\int_0^T\tr{\del{\frac{1}{2} \bM_t - \bfsigma_t\int_0^{\infty}e^{-2y\sqrt{-\bfsigma_t \dot{\bfSigma}^*_t \bfsigma_t}}  \dif y\bfsigma_t}\dot{\bfeta}_t}\dif t} \\
&= \E{\int_0^T\tr{\del{\frac{1}{2} \bM_t - \frac{1}{2}\bfsigma_t\del{\sqrt{-\bfsigma_t \dot{\bfSigma}_t^*\bfsigma_t}}^{-1}\bfsigma_t}\dot{\bfeta}_t}\dif t}.
\end{align*}
For the integrand to vanish, it suffices that
\begin{equation}
\label{1zwfZY3Cld}
\bM_t = \bfsigma_t\del{\sqrt{-\bfsigma_t \dot{\bfSigma}_t^*\bfsigma_t}}^{-1}\bfsigma_t, \qquad (t,\omega)\in [0,T]\times\Omega \quad \text{a.e.},
\end{equation}
but this is the same as \eqref{DzvEJ7XTQT}. Thus \eqref{foXZfSGcUD} holds. 

It remains to show that \eqref{DzvEJ7XTQT} is a global maximum, and not a saddle point or a local minimum. But this follows from the observation that $L(\dot{\bfSigma}_{[0,T]},\bM_{[0,T]})$ is concave in $\dot{\bfSigma}_{[0,T]}$ for each $\bM_{[0,T]} \in \cM$. To show this, take $\dot{\bfSigma}_{[0,T]}^1, \dot{\bfSigma}_{[0,T]}^2 \in \cA$ and $\lambda \in [0,1]$. Then, by the trace concavity of square root \cite[fact 3]{wangramdas}, we have for each $t \in [0,T]$
\begin{align}
\tr{\sqrt{-\bfsigma_t\del{\lambda\dot{\bfSigma}_t^1 + (1-\lambda)\dot{\bfSigma}_t^2}\bfsigma_t}} \geq \lambda\tr{\sqrt{-\bfsigma_t\dot{\bfSigma}_t^1\bfsigma_t}} + (1 - \lambda)\tr{\sqrt{-\bfsigma_t\dot{\bfSigma}_t^2\bfsigma_t}}.
\end{align}
Substituting this into \eqref{PC4t8KCzcX}, we see that
$$
L\!\del{\lambda\dot{\bfSigma}^1_{[0,T]} + (1-\lambda)\dot{\bfSigma}^2_{[0,T]},\bM_{[0,T]}} \geq \lambda L\!\del{\dot{\bfSigma}^1_{[0,T]},\bM_{[0,T]}} + (1-\lambda) L\!\del{\dot{\bfSigma}^2_{[0,T]},\bM_{[0,T]}}.
$$
Finally, substitution of \eqref{1zwfZY3Cld} into the Lagrangian \eqref{PC4t8KCzcX} produces the dual function \eqref{YTF6qPCOFy}.
\end{proof}

\begin{proof}[Proof of Proposition~\ref{FOHnj1xlta}]
Let $\bM^{1}_{[0,T]}, \bM^{2}_{[0,T]} \in \cM$ and $\lambda \in [0,1]$. Evaluating $g$ at $\lambda \bM^{1}_{[0,T]} + (1-\lambda)\bM^{2}_{[0,T]}$ gives
\begin{equation}
\label{AXS7rTZShb}
\frac{1}{2}\tr{\bC \del{\lambda\bM^1_0 + (1-\lambda)\bM^2_0}} + \frac{1}{2}\E{\int_0^T \tr{\bfsigma^2_t \sbr{\lambda\bM^1_t+(1-\lambda)\bM^2_t}^{-1}}\dif t}.
\end{equation}
Since $\bA \mapsto \tr{\bB\bA^{-1}\bB}$ is convex, we have
\begin{equation*}
\tr{\bfsigma_t \sbr{\lambda\bM^1_t+(1-\lambda)\bM^2_t}^{-1}\bfsigma_t} \leq \lambda \tr{\bfsigma_t\sbr{\bM_t^1}^{-1}\bfsigma_t} + (1-\lambda)\tr{\bfsigma_t\sbr{\bM_t^2}^{-1}\bfsigma_t}.
\end{equation*}
This means the expected value in \eqref{AXS7rTZShb} is bounded above by
\begin{equation}
\lambda\E{\int_0^T \tr{\bfsigma^2_t \sbr{\bM^1_t}^{-1}}\dif t} + (1-\lambda)\E{\int_0^T \tr{\bfsigma^2_t \sbr{\bM^2_t}^{-1}}\dif t},
\end{equation}
and so
$$
g\!\del{\lambda \bM^{1}_{[0,T]} + (1-\lambda)\bM^{2}_{[0,T]}} \leq \lambda g\!\del{\bM^{1}_{[0,T]}} + (1-\lambda)g\!\del{\bM^{2}_{[0,T]}}.
$$
\end{proof}

\begin{proof}[Proof of Theorem~\ref{hodau0V7QJ}]
Denote the value in \eqref{RUHFk3mIp0} as $L$, and note that $L \geq 0 > -\infty$. 

We first prove existence of optimizer for \eqref{RUHFk3mIp0}. Let $\alpha>0$ be such that $\frac{1}{\alpha}\bI\leq \bfsigma_t\leq \alpha\bI $ for all $t\in [0,T]$, almost surely. For the choice of the constant matrix $\bM_t=\bI$ we see that 
\begin{align*}
L\leq g(\bI)\leq  \frac{\alpha^2Tn+\tr{\bC}}{2}.\displaybreak
\end{align*}

For each $k$, let $\bM^k$ be an $\cS^n_{++}$ valued martingale satisfying 
$$g(\bM^k)\leq L+\frac{1}{k}\leq   \frac{\alpha^2Tn+\tr{\bC}}{2}+1.$$
The non-degeneracy of $\bC$ and the bound on $g(\bM^k)$ allows us to claim that $(\bM^k_0)_{k\geq 0}$ is bounded. Thus, without loss of generality we can assume that  $\bM^k_0\to \bM^{\infty}_0\in \cS^{n}_{+}$ as $k \to \infty$.

We now use a Komlós-type argument to construct a martingale sequence $\tilde{\bM}^{k}$ from $\bM^{k}$ with the property that $\tilde{\bM}^{k}$ converges almost surely to a component-wise limit $\bM^{\infty}$. Specifically, we iterate Delbaen--Schachermayer's compactness theorem across $n(n+1)/2$ scalar coordinates (the diagonal entries $\be_i^\top \bM\be_i$ and the off-diagonal directions $(\be_i+\be_j)^\top \bM(\be_i+\be_j)$ for $i<j$), applying each round's convex combinations to the underlying matrices so that coordinates which have already converged in earlier rounds remain convergent. Denote $\be_i = (0,\dots,0,1,0,\dots,0)^\top$, with $1$ being at ith position for $i = 1,\dots,n$. By the martingale representation theorem, there exists some process $H^{11,k} = (H^{11,k}_t)_{0 \leq t \leq T}$ for which
$$
\be_1^\top\bM_0^k\be_1 +  \int_0^t H^{11,k}_{s}\dif W_s = \be_1^\top \bM_t^k \be_1 \geq  0.
$$
Thus $\int_0^t H^{11,k}_{s}\dif W_s$ is uniformly bounded below by $-\be_1^\top\bM_0^k\be_1$, which by the boundedness of $(\bM^k_0)_{k\geq 0}$ established above is bounded by a deterministic constant independent of $k$. \cite[1.10 Theorem D]{delbaen1999compactness} asserts that there are finite convex combinations
$$
\tilde{H}^{11,k} \in {\rm{conv}}\{H^{11,k}, H^{11,k+1},\dots\}
$$
and a super-martingale $V_t^{11}$ for which
\begin{align}
\lim_{\substack{s\downarrow t \\ s\in \dQ_+}} \lim_{k\to\infty} \int_0^s \tilde{H}^{11,k}_r \dif W_r &= V^{11}_t \\
V^{11}_0 &\leq 0.
\end{align}

Applying these same convex combinations to the matrices $\bM^{k}, \bM^{k+1},\dots$, we build the martingale sequence $\tilde{\bM}^{k}, \tilde{\bM}^{k+1} \dots$. Take this new sequence and repeat the same procedure on the second diagonal term $\be_2^{\top}\bM^k\be_2$ to obtain $(V_t^{22})$.  Notice that the limit $V^{11}_t$ of the first diagonal does not change through this procedure.

We continue in this way along all diagonal entries, denoting again the new sequences obtained by $(\tilde{\bM})$. For the off-diagonal entries, we consider instead $\del{\be_i + \be_j}^{\top}\tilde{\bM}^k\del{\be_i + \be_j}$ with $i\neq j$ and we denote these limits $V_t^{ij}$. Altogether, this process is repeated $n(n+1)/2$ times. Define the matrix process $\bM^{\infty}$ by
$$
\sbr{\bM^{\infty}_t}_{i,j} = 
\begin{cases}
V_t^{ii} & i = j\\
\frac{1}{2}\del{V_t^{ij} - V_t^{ii} - V_t^{jj}} & i \neq j
\end{cases}
$$
and notice that it coincides with the component-wise limit of $(\tilde\bM^k)_{k}$.

Let $\bx \in \dR^{n}$ and $s \leq t$. For each $k$,  $(\tilde{\bM}^k_t)_{0\le t\le T}$ is a martingale, and so
$$
\E{\bx^{\top}\tilde{\bM}_t^k \bx \middle| \sF_s} = \bx^{\top}\tilde{\bM}_s^k \bx.
$$
Taking $k \to \infty$, Fatou's lemma gives
$$
\E{\liminf_{k}\bx^{\top}\tilde{\bM}_t^k \bx \middle| \sF_s} \leq  \liminf_{k}\bx^{\top}\tilde{\bM}_s^k \bx.
$$
But a component-wise limit exists and is $(\bM^{\infty}_t)_{0\le t\le T}$ and so
\begin{equation}
\label{1U7mdK1n5U}
\E{\bx^{\top}\bM^{\infty}_t \bx \middle| \sF_s} \leq  \bx^{\top}\bM^{\infty}_s \bx.
\end{equation}
Since $\bx$ is arbitrary, this means that $\E{\bM^{\infty}_t  \middle| \sF_s} \leq  \bM^{\infty}_s.$

Thus $\bM^{\infty}$ is a supermartingale as defined by Definition~\ref{sySq817Omx}.

But \eqref{YTF6qPCOFy} is convex in $\bM$, so any supermartingale minimizer may be improved by taking only its local martingale component. We first note that $\bM^{\infty}_t > 0$ for a.e.\ $(t,\omega)$: otherwise the integrand $\tr{\bfsigma^2_t (\bM^\infty_t)^{-1}}$ in $g(\bM^\infty)$ would equal $+\infty$ on a set of positive measure, contradicting $g(\bM^\infty) \leq L < \infty$ (verified below).
The matrix-valued Doob--Meyer decomposition theorem (Theorem~\ref{YJoB18UgVC} and Remark~\ref{rem.supermart}) permits the representation $\bM^{\infty}_t = \bL_t + \bA_t$, for some local martingale $\bL_t$ and some nonincreasing process $\bA_t$ which starts at $\bf 0$. Since $\bM^{\infty}_t  \leq \bL_t$ for each $t$, both are positive definite a.e., and operator-monotonicity of inversion on $\cS^n_{++}$ yields $\del{\bL_t}^{-1} \leq \del{\bM^{\infty}_t}^{-1}$, and so
\begin{align*}
g\!\del{\bM^{\infty}} &= \frac{1}{2}\tr{\bC \bM^{\infty}_0} + \frac{1}{2}\E{\int_0^T \tr{\bfsigma^2_t \del{\bM_t^{\infty}}^{-1}}\dif t} \\
&\geq \frac{1}{2}\tr{\bC \bM^{\infty}_0} + \frac{1}{2}\E{\int_0^T \tr{\bfsigma^2_t \del{\bL_t}^{-1}}\dif t} = g\!\del{\bL}.
\end{align*}

We now show that $\bL$ is an optimizer of \eqref{RUHFk3mIp0}. By convexity of $g$ and the construction of $\tilde \bM^k$ as a convex combination of $\bM^j$ for $j\geq k$, we have 
$$L+\frac{1}{k}\geq g\!\del{\bM^{k}}\geq g\!\del{\tilde \bM^{k}}.$$
Since the integrand $\tr(\bfsigma^2_t \bM^{-1}_t)$ is non-negative and lower-semicontinuous in $\bM\in\cS^n_{++}$, Fatou's lemma (applied to the a.s.\ limit $\tilde\bM^k\to\bM^\infty$ and using continuity of $\bM\mapsto\bM_0$ in the boundary term) gives
$$g( \bM^{\infty})\leq\liminf_k g\!\del{\tilde\bM^{k}}\leq \liminf_k\del{L+\frac{1}{k}}=L.$$
Combined with $g(\bL)\leq g(\bM^\infty)$ from above,
$$g(\bL)\leq g( \bM^{\infty})\leq L,$$ 
which shows the optimality of $\bL$.
\pagebreak

We now assume that there are two optimizers denote   $\bM^{*,1},\bM^{*,2}$ and by abuse of notation define the convex function 
$$[0,1] \ni \lambda \mapsto g(\lambda):=g(\lambda \bM^{*,2}+(1-\lambda)\bM^{*,1})$$ which by convexity and optimality is constant along $\lambda \in [0,1]$. Thus, 
$$0=g''(0)=\E{\int_0^T \tr{\bfsigma_t \sbr{\bM_t^{*,1}}^{-1}(\bM_t^{*,1}-\bM_t^{*,2})\sbr{\bM_t^{*,1}}^{-1}(\bM_t^{*,1}-\bM_t^{*,2})\sbr{\bM_t^{*,1}}^{-1}\bfsigma_t}\dif t}$$
which implies $\bM_t^{*,1}=\bM_t^{*,2}$, $dt\times d\dP$-a.s. due to the positivity of the matrices involved. 
\end{proof}

\begin{proof}[Proof of Corollary~\ref{cor:strong-duality}]
Choose $\dot{\bfSigma}^*$ as in \eqref{Rwq3zFmuoG}. We will show that $\del{\dot{\bfSigma}^*,\bM^*}$ is a saddle point of $L$, which will prove \eqref{9CcesNg5w6}.

It suffices to show that
\begin{align}
\nabla_{\dot{\bfSigma}}L\!\del{\dot{\bfSigma}^*,\bM^*} &= {\bf 0} \label{PvbNNXYL8q} \\
\nabla_{\bM}L\!\del{\dot{\bfSigma}^*,\bM^*} &= {\bf 0}. \label{HHcMw53Xr9}
\end{align}
The identity \eqref{PvbNNXYL8q} reduces to
\begin{equation}
\label{BUSu1MNQ9m}
\bM^*_t = \bfsigma_t\del{\sqrt{-\bfsigma_t \dot{\bfSigma}_t^*\bfsigma_t}}^{-1}\bfsigma_t, \qquad (t,\omega)\in [0,T]\times\Omega \quad \text{a.e.},
\end{equation}
as was shown in the proof of Proposition~\ref{wn0NtKxptx}, and this is satisfied by \eqref{Rwq3zFmuoG}. On the other hand, \eqref{HHcMw53Xr9} reduces to
$$
\bC + \int_0^T \dot{\bfSigma}^*_t \dif t = {\bf 0},
$$
which, due to our choice of $\dot{\bfSigma}^*_t$, becomes
\begin{equation}
\label{ukQnmQso9N}
\bC - \int_0^T \sbr{\bM_t^*}^{-1}\bfsigma_t^2\sbr{\bM_t^*}^{-1} \dif t = {\bf 0}.
\end{equation}
This is just the statement $\nabla_{\bM}g(\bM^*) = {\bf 0}$, where $g$ is the dual function \eqref{YTF6qPCOFy}. By the convexity of $g$ (Proposition~\ref{FOHnj1xlta}) and the optimality of $\bM^*$ in Theorem~\ref{hodau0V7QJ}, this holds too.
\end{proof}

\begin{proof}[Proof of Proposition~\ref{TUpmwPhTpr}]
Fix $\Sigma\in\cS^{n}_{++}$, $A\in\cS^{n}$, and $\e\geq 0$ such that $\Sigma+\e A\in\cS^{n}_{++}$. By Theorem~\ref{hodau0V7QJ}, the optimizer $\bM^{*,0,\Sigma+\e A}$ exists and gives
\[
G^*(\Sigma+\e A) = \frac{1}{2}\tr{(\Sigma+\e A)\,\bM^{*,0,\Sigma+\e A}_0} + \frac{1}{2}\E{\int_0^T \tr{\bfsigma^2_s\,(\bM^{*,0,\Sigma+\e A}_s)^{-1}}\dif s}.
\]
For $\e>0$, the suboptimality of $\bM^{*,0,\Sigma+\e A}$ at initial covariance $\Sigma$ and of $\bM^{*,0,\Sigma}$ at initial covariance $\Sigma+\e A$ gives the chain
\begin{align*}
\frac{1}{2}\tr{A\,\bM^{*,0,\Sigma+\e A}_0}
&= \frac{1}{\e}\sbr{G_{\Sigma+\e A}\!\del{\bM^{*,0,\Sigma+\e A}} - G_\Sigma\!\del{\bM^{*,0,\Sigma+\e A}}}\\
&\leq \frac{G^*(\Sigma+\e A) - G^*(\Sigma)}{\e}\\
&\leq \frac{1}{\e}\sbr{G_{\Sigma+\e A}\!\del{\bM^{*,0,\Sigma}} - G_\Sigma\!\del{\bM^{*,0,\Sigma}}}
= \frac{1}{2}\tr{A\,\bM^{*,0,\Sigma}_0},
\end{align*}
where the slippage integrals cancel within each bracket. Thus
\begin{align}\label{eq:boundder}
\frac{1}{2}\tr{A\,\bM^{*,0,\Sigma+\e A}_0}
\;\leq\;
\frac{G^*(\Sigma+\e A) - G^*(\Sigma)}{\e}
\;\leq\;
\frac{1}{2}\tr{A\,\bM^{*,0,\Sigma}_0}.
\end{align}

We next show $\bM^{*,0,\Sigma+\e_n A}_0\to\bM^{*,0,\Sigma}_0$ along every sequence $\e_n\downarrow 0$. Suppose for contradiction this fails along some such sequence. If $\bM^{*,0,\Sigma+\e_n A}_0$ is unbounded, then since $\Sigma+\e_n A\to\Sigma> 0$, the boundary term $\frac{1}{2}\tr{(\Sigma+\e_n A)\,\bM^{*,0,\Sigma+\e_n A}_0}$ blows up along a subsequence, contradicting the uniform bound $G^*(\Sigma+\e_n A)\leq G_{\Sigma+\e_n A}(\bI)\to G_\Sigma(\bI)<\infty$. So along a subsequence (not relabelled), $\bM^{*,0,\Sigma+\e_n A}_0\to M$ for some $M\neq\bM^{*,0,\Sigma}_0$.

Proceeding as in the proof of Theorem~\ref{hodau0V7QJ}, convex combinations of the processes $\del{\bM^{*,0,\Sigma+\e_n A}_s}_{s\in[0,T]}$ converge to an $\sF^{\bW}$-local martingale $\bM^{\infty}$ with $\bM^{\infty}_0=M$. Since $G_\Sigma$ is lower-semicontinuous along such convex combinations and $G_{\Sigma+\e_n A}\to G_\Sigma$ pointwise as $\e_n\downarrow 0$, the limit $\bM^{\infty}$ attains the infimum defining $G^*(\Sigma)$. Uniqueness in Theorem~\ref{hodau0V7QJ} forces $\bM^{\infty}=\bM^{*,0,\Sigma}$, contradicting $M\neq\bM^{*,0,\Sigma}_0$.

Continuity of $\Sigma\mapsto\bM^{*,0,\Sigma}_0$ now lets the bounds in \eqref{eq:boundder} pinch as $\e\downarrow 0$, giving the directional derivative
\[
D_A G^*(\Sigma) = \frac{1}{2}\tr{A\,\bM^{*,0,\Sigma}_0}.
\]
Linearity in $A$ yields \eqref{eq:G-star-grad}, and continuity of $D_\Sigma G^*$ follows from that of $\Sigma\mapsto\bM^{*,0,\Sigma}_0$.
\end{proof}

\begin{proof}[Proof of Corollary~\ref{cor:V-diff}]
For each $\Sigma\in\cO$, applying Corollary~\ref{cor:strong-duality} with $\Sigma$ in place of $\bC$ yields $V(0,\Sigma)=G^*(\Sigma)$. Hence $V(0,\cdot)\equiv G^*$ on $\cO$, and the conclusion follows from Proposition~\ref{TUpmwPhTpr}.
\end{proof}

\section{Proofs for the equilibrium theorem}
\label{app:equilibrium-proofs}

\subsection{Properties of the equilibrium}
In this subsection, we assume that the informed trader uses the strategy \eqref{ud3k6AlV7q}.

\begin{proposition}
\label{v0tMLvYb0E}
The equilibrium price process in \eqref{2vGs8ogSts} is
\begin{equation}
\label{SWjcJYb4sK}
\bP^*_t = \bfH^{*}(t,\bfxi^*_t) = \E{\vt \middle| \sF^{M}_t}= \bp_0 + \underbrace{\int_0^t \sbr{\bM^*_s}^{-1} \dif \bY^*_s}_{\bfxi^*_t}
\end{equation}
and the market maker's optimal filtering error is
\begin{align}
{\rm{Var}}\!\del{\vt\middle|\sF^{M}_t} &= \bfSigma^*_t = \bC - \int_0^t \sbr{\bM^*_s}^{-1} \bfsigma^2_s \sbr{\bM^*_s}^{-1} \dif s \label{ENZWGUDa1A} \\
&= \int_t^T \sbr{\bM^*_s}^{-1} \bfsigma^2_s \sbr{\bM^*_s}^{-1} \dif s. \label{53QrAITZQR}
\end{align}
\end{proposition}
Notice that the integral \eqref{53QrAITZQR} is $\sF^{M}_t$ measurable, even though the integral in time ``goes into the future.''

\begin{proof}
This is an application of the Kalman--Bucy filter \cite[theorem 12.7]{liptser2013statistics}. The hidden process is the static random variable $\vt$, and the observable process is the pair $\del{\bfxi_t,\bfsigma_t}$, which by the strategy \eqref{ud3k6AlV7q} evolves according to
\begin{align}
\dif\bfxi^*_t &= \sbr{\bM^*_t}^{-1}\bfsigma_t^{2}\sbr{\bM^*_t}^{-1}\sbr{\bfSigma^*_t}^{-1}\del{\vt - \bp_0 - \bfxi^*_t}\dif t + \sbr{\bM^*_t}^{-1}\bfsigma_t\dif\bB_t, \quad  \bfxi^*_0 = {\bf 0}\\
 \bfsigma_t &\in \sF^{W}_t\subset \sF^{M}_t.
\end{align}

Note that
\begin{align*}
\int_0^T \sum_{i,j}^{n}\del{\sbr{\bM^*_t}^{-1}\bfsigma_t}^2_{ij}\dif t &= \int_0^T\tr{\sbr{\bM^*_t}^{-1}\bfsigma_t^2 \sbr{\bM^*_t}^{-1}}\dif t \\ &= -\int_0^T\tr{\dot{\bfSigma}^*_t}\dif t = \tr{\bC} < \infty
\end{align*}

The filtering mean $\bm_t = \E{\vt \middle| \sF^{M}_t}$ satisfies the SDE
\begin{align}
\dif\bm_t &= \bfgamma_t\sbr{\bfSigma^*_t}^{-1}\del{\dif\bfxi_t - \sbr{\bM^*_t}^{-1}\bfsigma_t^{2}\sbr{\bM^*_t}^{-1}\sbr{\bfSigma^*_t}^{-1}\del{\bm_t - \bp_0 - \bfxi^*_t}} \label{E0ZBEJjXcw} \\
\bm_0 &= \bp_0, 
\end{align}
and the filtering variance $\bfgamma_t = {\rm{Var}}\!\del{\vt\middle|\sF^{M}_t}$ satisfies the Riccati equation
\begin{align}
\dot{\bfgamma}_t &= - \bfgamma_t \sbr{\bfSigma^*_t}^{-1}\sbr{\bM^*_t}^{-1}\bfsigma_t^{2}\sbr{\bM^*_t}^{-1}\sbr{\bfSigma^*_t}^{-1}\bfgamma_t \label{5bbt09Ww8Q} \\
\bfgamma_0 &= \bC.
\end{align}
Now, in light of \eqref{DzvEJ7XTQT}, we see that $\bfgamma_t = \bfSigma^*_t$ is a solution to \eqref{5bbt09Ww8Q}. This gives us \eqref{ENZWGUDa1A}. Updating \eqref{E0ZBEJjXcw} with this selection, it becomes clear that $\bP_t$ as it appears in \eqref{SWjcJYb4sK} is a solution to \eqref{E0ZBEJjXcw}.
\end{proof}

\begin{proposition}
\label{rdMUtjXwKx}
Let $\bfxi^*$ and $\bM^*$ be defined as in \eqref{CrLPP8jKQG} and \eqref{SWjcJYb4sK}. Then,
\begin{equation}
\hat{\bB}_t = \int_0^t \bfsigma_s^{-1} \bM^{*}_s \dif \bfxi^*_s\label{ULKc9KgyLV}
\end{equation}
is an $n$-dimensional Brownian motion with respect to the market maker's filtration $(\sF^M_t)_{0\le t\le T}$.
\end{proposition}

\begin{proof}
We invoke L\'evy's characterization. First, observe that \eqref{ULKc9KgyLV} is continuous. Next, compute
\begin{align*}
\dif \,\langle\hat{\bB}\rangle_t &= \bfsigma_t^{-1}\bM^*_t \dif\,\langle\bfxi^*\rangle_t\bM^*_t\bfsigma_t^{-1} \\
&= \bfsigma_t^{-1}\bM^*_t \sbr{\bM^*_t}^{-1} \bfsigma_t \dif\,\langle\bB\rangle_t \bfsigma_t\sbr{\bM^*_t}^{-1}\bM^*_t\bfsigma_t^{-1} = \bI \dif t.
\end{align*}
Finally, note that the processes $(\bfsigma_t)_{0\le t\le T}$, $(\bM^*_t)_{0\le t\le T}$, and $(\bfxi^*_t)_{0\le t\le T}$ are all adapted to the filtration $(\sF^M_t)_{0\le t\le T}$.
\end{proof}

\begin{proposition}
\label{cnTt7K54xR}
Conditioned on $\sF^{M}_t$, $\bfxi^*_T$ is Gaussian with mean $\bfxi^*_t$ and variance $$\bfSigma_t^* = \bC - \int_0^t\sbr{\bM^*_s}^{-1}\bfsigma^2_s\sbr{\bM^*_s}^{-1}\dif s = \int_t^T\sbr{\bM^*_s}^{-1}\bfsigma^2_s\sbr{\bM^*_s}^{-1}\dif s.$$
\end{proposition}

\begin{proof}
Fix $\bu\in\dR^n$. By It\^o's lemma, we have
$$
\dif e^{i\bu^{\top}\bfxi^*_t} = ie^{i\bu^{\top}\bfxi^*_t}\bu^{\top}\dif\bfxi^*_t - \frac{1}{2}e^{i\bu^{\top}\bfxi^*_t}\bu^{\top}\dif\,\langle\bfxi^*\rangle_t\bu,
$$
or, in light of \eqref{rdMUtjXwKx},
$$
\dif e^{i\bu^{\top}\bfxi^*_t} = ie^{i\bu^{\top}\bfxi^*_t}\bu^{\top}\sbr{\bM^*_t}^{-1}\bfsigma_t\dif\hat{\bB}_t - \frac{1}{2}e^{i\bu^{\top}\bfxi^*_t}\bu^{\top}\sbr{\bM^*_t}^{-1}\bfsigma_t^2\sbr{\bM^*_t}^{-1}\bu \dif t.
$$
We integrate from $t$ to $T$ to obtain
$$
e^{i\bu^{\top}\bfxi^*_T} - e^{i\bu^{\top}\bfxi^*_t} = \int_t^T ie^{i\bu^{\top}\bfxi^*_s}\bu^{\top}\sbr{\bM^*_s}^{-1}\bfsigma_s\dif\hat{\bB}_s - \frac{1}{2} \int_t^T e^{i\bu^{\top}\bfxi^*_s}\bu^{\top}\sbr{\bM^*_s}^{-1}\bfsigma_s^2\sbr{\bM^*_s}^{-1}\bu \dif s.
$$
Applying $\sF^M_t$ conditional expectation, the $\dif\hat{\bB}$ integral vanishes because
$$
\E{\int_0^T\bu^{\top}\sbr{\bM^*_s}^{-1}\bfsigma_s^2\sbr{\bM^*_s}^{-1}\bu \dif s} = \bu^{\top}\bC\bu < \infty,
$$
and we are left with
$$
\E{e^{i\bu^{\top}\bfxi^*_T}\middle|\sF^M_t} - e^{i\bu^{\top}\bfxi^*_t} = - \frac{1}{2} \int_t^T \E{e^{i\bu^{\top}\bfxi^*_s}\middle|\sF^M_t}\bu^{\top}\sbr{\bM^*_s}^{-1}\bfsigma_s^2\sbr{\bM^*_s}^{-1}\bu \dif s.
$$
This is an integral equation whose solution gives
\[
\E{e^{i\bu^{\top}\bfxi^*_T}\middle|\sF^M_t} = \exp\!\del{i\bu^{\top}\bfxi^*_t - \frac{1}{2}\bu^{\top}\del{\int_t^T\sbr{\bM^*_s}^{-1}\bfsigma_s^2\sbr{\bM^*_s}^{-1} \dif s}\bu}.\qedhere
\]
\end{proof}

\subsection{The wealth decomposition}
\begin{definition}
A trading strategy $\bX$ is said to be {\it inconspicuous} provided that it is absolutely continuous and that, for all $t \in [0,T]$,
$$
\E{\dod{\bX_t}{t}\middle| \sF^M_t} = 0.
$$
\end{definition}

\begin{proposition}
The strategy \eqref{ud3k6AlV7q} is inconspicuous.
\end{proposition}

\begin{proof}
Clearly \eqref{ud3k6AlV7q} is absolutely continuous. Next, compute
\begin{align*}
\E{\dod{\bX^*_t}{t}\middle| \sF^M_t} &= \E{\bfsigma_t^{2}\sbr{\bM^*_t}^{-1}\bfSigma_t^{-1}\del{\vt - \bP_t^*}\middle| \sF^M_t} \\
&= \bfsigma_t^{2}\sbr{\bM^*_t}^{-1}\bfSigma_t^{-1}\del{\E{\vt\middle| \sF^M_t} - \bP_t^*} = 0.
\end{align*}
The quantity in the parenthesis vanishes due to \eqref{SWjcJYb4sK}.
\end{proof}

\begin{definition}
Define $V \colon \dR^+ \times \mathcal{S}^n_{++} \mapsto \dR$ the dynamic programming value of \eqref{OxpDCjjac4}  by
\begin{equation}
\label{Du9L08iE1l}
V(t,\bC) = \sup_{\dot{\bfSigma} \leq \bf{0}}\cbr{\E{\int_t^T \tr{\sqrt{-\bfsigma_s \dot{\bfSigma}_s \bfsigma_s}}\dif s\middle| \sF^{\bW}_t} \colon \bC + \int_t^T \dot{\bfSigma}_s \dif s = 0}.
\end{equation}
\end{definition}

\begin{lemma}
\label{kqqCZ0K8sT}
Suppose there exists a maximizer $(\dot{\bfSigma}^*_t)_{0\le t\le T}$ to the primal problem \eqref{OxpDCjjac4}, and denote by $(\bfSigma^*_t)_{0\le t\le T}$ its optimal trajectory. Then, $V(t,\bfSigma^*_t)$ admits the dynamics
\begin{equation}
\label{36ra03OLTA}
\dif V(t,\bfSigma^*_t) = -\tr{\bfsigma_t[\bM^*_t]^{-1}\bfsigma_t} \dif t + \bR(t,\bfSigma^*_t)^{\top}\dif \bW_t
\end{equation}
for some adapted square integrable random field $\del{\bR(t,\bA)}_{0 \leq t \leq T}$.
\end{lemma}

\begin{proof}
If $(\bfSigma^*_t)_{0\le t\le T}$ denotes the optimal trajectory, then we have from \eqref{Du9L08iE1l} the identity
$$
V(t,\bfSigma^*_t) = \E{\int_t^T \tr{\sqrt{-\bfsigma_s \dot{\bfSigma}^*_s \bfsigma_s}}\dif s \middle| \sF^{\bW}_t}.
$$
We add to both sides the $\sF^{\bW}_t$-measurable quantity $\int_0^t \tr{\sqrt{-\bfsigma_s \dot{\bfSigma}^*_s \bfsigma_s}}\dif s$ and thus obtain
\begin{equation}
\label{VVo3OClH8F}
V(t,\bfSigma^*_t) + \int_0^t \tr{\sqrt{-\bfsigma_s \dot{\bfSigma}^*_s \bfsigma_s}}\dif s = \E{\int_0^T \tr{\sqrt{-\bfsigma_s \dot{\bfSigma}^*_s \bfsigma_s}}\dif s \middle| \sF^{\bW}_t}.
\end{equation}
The RHS is a martingale and thus admits the representation
$$
\E{\int_0^T \tr{\sqrt{-\bfsigma_s \dot{\bfSigma}^*_s \bfsigma_s}}\dif s \middle| \sF^{\bW}_t} = \int_0^t \bR(s,\bfSigma^*_s)^{\top}\dif \bW_s
$$ for some adapted square integrable $(\bR(t,\bA))_{0\le t\le T}$. The integral which appears on the RHS may be written instead as
$$
\int_0^t \tr{\sqrt{-\bfsigma_s \dot{\bfSigma}^*_s \bfsigma_s}}\dif s = \int_0^t\tr{\bfsigma_s[\bM^*_s]^{-1}\bfsigma_s} \dif s
$$
due to the relation \eqref{DzvEJ7XTQT}.
Thus, \eqref{36ra03OLTA} is just the identity \eqref{VVo3OClH8F}.
\end{proof}

\begin{theorem}
\label{thm:insider-value}
Suppose the data $(\bC,(\bfsigma_t)_{0\le t\le T})$ satisfy the martingale dual condition (Definition~\ref{def:MDC}), and let $(\bM^*_t)_{0\le t\le T}$ be the resulting $\sF^{\bW}$-martingale optimizer of \eqref{RUHFk3mIp0}. Let $(\dot{\bfSigma}^*_t)_{0\le t\le T}$ be defined as in \eqref{PvbNNXYL8q}. Define the random potentials $\varphi \colon \Omega \times \dR^+ \times \dR^{n} \mapsto \dR$ and $\psi \colon \Omega \times\dR^+ \times \dR^{n} \times \dR^{n} \mapsto \dR$ by
\begin{align}
\varphi(\omega, t, \bfxi) &= \frac{\del{\bfxi+\bp_0}^{\top}\bM^*_t(\omega)\del{\bfxi+\bp_0}}{2} + \frac{1}{2} V(t,\bfSigma^*_t) \\
\psi^{\vt}(\omega, t, \bp) &= \frac{\del{\vt-\bp}^{\top}\bM^*_t(\omega) \del{\vt-\bp}}{2} + \frac{1}{2}V(t, \bfSigma^*_{t}).
\end{align}
If the market maker sets prices by \eqref{CrLPP8jKQG}-\eqref{2vGs8ogSts}, then the informed trader's $\vt$ conditional expected profit is at most
\begin{equation}
\label{BPMSIqokAh}
\dE\!\sbr{\int_0^T\del{\vt - \bP_t}^{\top} \dif \bX_t -\langle\bX,\bP\rangle_T\middle| \sF^I_0} \leq \dE\!\sbr{ \psi^{\vt}(0, \bp_0) \middle| \sF^I_0}
\end{equation}
from time $0$, and in general at most
\begin{equation}
\label{FTemha9Y4r}
\dE\!\sbr{\int_t^T\del{\vt - \bP_s}^{\top} \dif \bX_s -\langle\bX,\bP\rangle_T + \langle\bX,\bP\rangle_t\middle| \sF^I_t} \leq \dE\!\sbr{ \psi^{\vt}(t, \bP_t) \middle| \sF^I_t}
\end{equation}
from time $t \in [0,T]$. 

Furthermore, these upper bounds are attained when $\bX_t$ is inconspicuous and ensures $\bfxi_T = \vt - \bp_0$ a.s.

\end{theorem}

\begin{proof}
Let $\bY_t = \bX_t + \bZ_t$ denote the order flow from any admissible strategy $\bX_t$. By It\^o's lemma, we have
\begin{equation}
\label{neiI6IKpfa}
\begin{split}
\dif \varphi(t,\bfxi_t) &= \del{\bfxi_t + \bp_0}^{\top}\bM^*_t \dif \bfxi_t + \frac{1}{2}\del{\bfxi_t + \bp_0}^{\top}\dif\bM^*_t\del{\bfxi_t + \bp_0} \\ & \qquad + \frac{1}{2}\dif \bfxi_t^{\top} \bM^*_t\dif \bfxi_t + \del{\bfxi_t + \bp_0}^{\top}\dif\bM^*_t\dif \bfxi_t + \frac{1}{2}\dif V(t,\bfSigma^*_t).
\end{split}
\end{equation}
Observe that $\dif\bM^*_t\dif \bfxi_t = {\bf 0}$ because $\dif\bM^*_t$ and $\dif \bfxi_t$ are comprised of independent Brownian differentials ($\dif W_t$ for $\bM^*$, $\dif B_t$ for $\bfxi$). In light of \eqref{CrLPP8jKQG}-\eqref{2vGs8ogSts}, we may replace $\bfxi_t + \bp_0$ by the price process $\bP_t$. We have $\bM^*_t\dif\bfxi_t = \dif\bY_t$ due to \eqref{SWjcJYb4sK}. These considerations taken into account, \eqref{neiI6IKpfa} updates to
\begin{equation}
\label{QZr6y2kFvU}
\dif \varphi(t,\bfxi_t) = \bP_t^\top \dif \bY_t + \frac{1}{2}\bP_t^\top\dif\bM^*_t\bP_t + \frac{1}{2}\dif\bY_t^\top \bfLambda^*_t\dif\bY_t + \frac{1}{2}\dif V(t,\bfSigma^*_t).
\end{equation}
The third RHS term is
\begin{align*}
\dif\bY_t^\top \bfLambda^*_t\dif\bY_t &= \del{\dif \bX_t + \dif \bZ_t}^\top \bfLambda^*_t\del{\dif \bX_t + \dif \bZ_t} \\
&= \dif \bX_t^\top\bfLambda^*_t\dif \bX_t + 2 \dif \bX_t^\top\bfLambda^*_t\dif \bZ_t + \dif \bZ_t^\top\bfLambda^*_t\dif \bZ_t \\
&= \dif \bX_t^\top\bfLambda^*_t\dif \bX_t + 2 \dif \bX_t^\top\bfLambda^*_t \del{\dif\bY_t - \dif\bX_t} + \dif \bZ_t^\top\bfLambda^*_t\dif \bZ_t \\
&= -\dif \bX_t^\top\bfLambda^*_t\dif \bX_t + 2 \dif \bX_t^\top\bfLambda^*_t \dif\bY_t + \dif \bZ_t^\top\bfLambda^*_t\dif \bZ_t \\
&= -\dif \bX_t^\top\bfLambda^*_t\dif \bX_t + 2 \dif\,\langle\bX,\bP\rangle_t  + \tr{\bfsigma_t \bfLambda^*_t \bfsigma_t}\dif t,
\end{align*}
so again \eqref{QZr6y2kFvU} becomes, with the help of Lemma~\ref{kqqCZ0K8sT},
\begin{equation}
\label{bjfFszgqRk}
\dif \varphi(t,\bfxi_t) = \bP_t^\top \dif \bY_t + \frac{1}{2}\bP_t^\top\dif\bM^*_t\bP_t - \frac{1}{2}\dif \bX_t^{\top}\bfLambda^*_t\dif \bX_t + \dif\,\langle\bX,\bP\rangle_t + \frac{1}{2}\bR_t^{\top}\dif \bW_t.
\end{equation}
Similar calculations with $\vt^{\top}\bM_t\del{\bfxi_t+\bp_0}$ produces
\begin{equation}
\label{2YgjW0Cluj}
\dif\del{\vt^{\top}\bM_t\del{\bfxi_t+\bp_0}} = \vt^{\top}\dif \bM^*_t\bP_t + \vt^{\top}\dif\bY_t.
\end{equation}
Subtracting \eqref{bjfFszgqRk} from \eqref{2YgjW0Cluj}, we find
\begin{equation}
\label{3VluS9QDeK}
\begin{split}
\del{\vt - \bP_t}^{\top} \dif \bY_t  - \dif \, \langle \bX,\bP \rangle_t &= \dif\del{\vt^{\top}\bM_t\del{\bfxi_t+\bp_0} - \varphi(t,\bfxi_t)} + \frac{1}{2}\bR_t^{\top}\dif \bW_t\\ & \qquad + \del{\frac{1}{2}\bP_t - \vt}^{\top}\dif \bM^*_t \bP_t^* - \frac{1}{2}\dif\bX_t^{\top}\bfLambda^*_t\dif\bX_t
\end{split}
\end{equation}
If we integrate \eqref{3VluS9QDeK} over $t \in [0,T]$ and apply $\sF^I_0$ conditional expectation, then $\dif \bY_t = \dif \bX_t + \dif \bZ_t$ reduces to $\dif \bX_t$ on the left-hand side. On the right-hand side, both stochastic integrals have vanishing $\sF^I_0$-conditional expectation. The term $\frac{1}{2}\int_0^T\bR_t^{\top}\dif \bW_t$ is a true martingale because $\bR$ is square integrable (Lemma~\ref{kqqCZ0K8sT}). The term $\int_0^T\del{\frac{1}{2}\bP_t - \vt}^{\top}\dif \bM^*_t\,\bP^*_t$ is a true martingale because, under MDC, $\bM^*$ is a true $\sF^{\bW}$-martingale and its integrand is square integrable, using $\vt\in L^2$, the square integrability of $\bP$ and $\bP^*$ established in the proof of Theorem~\ref{I3YrvmnwLc}, and the bounds on $\bfsigma$. Had $\bM^*$ been only a local martingale, this integral could fail to have zero mean and the identity below would weaken to an inequality. We are left with 
\begin{equation}
\begin{split}
&\E{\int_0^T\del{\vt - \bP_t}^{\top} \dif \bX_t - \langle \bX,\bP \rangle_T\middle|\sF^I_0} = \\
&\quad \E{\vt^{\top}\bM_T\del{\bfxi_T+\bp_0} - \varphi(T,\bfxi_T)   - \vt^{\top}\bM_0\bp_0 + \varphi(0,{\bf 0})   - \int_0^T \frac{1}{2}\dif\bX_t^{\top}\bfLambda^*_t\dif\bX_t\middle| \sF^I_0}.
\end{split}
\end{equation}
Since $\varphi(T,\bfxi_T)$ is quadratic in $\bfxi_T$, the RHS terms which involve $\bfxi_T$ achieve their maximum at $\bfxi_T = \vt - \bp_0$. Furthermore, $-\int_0^T\dif\bX_t^{\top}\bfLambda^*_t\dif\bX_t$ is at most $0$, and this occurs when $\bX_t$ is absolutely continuous. Altogether, we have the upper bound
\begin{equation}
\begin{split}
& \dE\!\sbr{\int_0^T\del{\vt - \bP_t}^{\top} \dif \bX_t - \langle \bX,\bP \rangle_T\middle|\sF^I_0} \\
& \qquad \leq \dE\Bigg[\frac{\del{\vt-\bp_0}^{\top}\bM^*_0\del{\vt-\bp_0}}{2} + \frac{1}{2}V(0,\bC)\Bigg| \sF^I_0\Bigg]
\end{split}
\end{equation}
which is \eqref{BPMSIqokAh}. To obtain \eqref{FTemha9Y4r}, integrate \eqref{3VluS9QDeK} over $[t,T]$ instead and apply $\sF^I_t$-conditional expectations.
\end{proof}

\subsection{Proof of Theorem~\ref{I3YrvmnwLc}}
\label{GmYBHbBEpZ}
\begin{proof}
First, fix the trading strategy \eqref{ud3k6AlV7q}. By Proposition~\ref{v0tMLvYb0E}, \eqref{CrLPP8jKQG} and \eqref{2vGs8ogSts} are rational in the sense of Definition~\ref{def:rational}. Furthermore, notice that $\bp_0 + \bfxi = \nabla_{\bfxi}\del{\bp_0\top\bfxi + \frac{1}{2}\bfxi^\top\bfxi}$, so \eqref{2vGs8ogSts} coincides with the gradient of a convex function. To check the integrability assumptions of Definition~\ref{def:pricingrule}, note that
\begin{align*}
\E{\norm{\bfH\!\del{T,\bfxi^*_T}}^2\middle| \sF^M_0} &= \E{\norm{\bp_0 + \bfxi^*_T}^2\middle| \sF^M_0}  \\
&= \norm{\bp_0}^2 + \tr{\var{\bfxi^*_T\middle|\sF^M_0}} = \norm{\bp_0}^2 + \tr{\bC} < \infty.
\end{align*}
The last equality follows from Proposition~\ref{cnTt7K54xR}.

Similarly,
\begin{align}
\E{\int_0^T \norm{\bfH\!\del{t,\bfxi^*_t}}^2 \dif t \,\middle| \sF^M_0} &= 
\E{\int_0^T \norm{\bp_0 + \bfxi_t^*}^2 \dif t \,\middle| \sF^M_0}. \label{YTsBhTHpDV}
\end{align}
Appealing again to Proposition~\ref{cnTt7K54xR} and Jensen's inequality,
$$
\norm{\bp_0 + \bfxi_t^*}^2 \leq \E{\,\norm{\bp_0 + \bfxi_T^*}^2 \middle| \sF^M_t}
$$
and so \eqref{YTsBhTHpDV} is bounded above by
\begin{align*}
\E{\int_0^T \E{\,\norm{\bp_0 + \bfxi_T^*}^2 \middle| \sF^M_t} \dif t \middle| \sF^M_0} &= T \E{\,\norm{\bp_0 + \bfxi_T^*}^2  \middle| \sF^M_0} \\
&=T\del{\norm{\bp_0}^2 + \tr{\bC}} < \infty
\end{align*}
Thus, we have established that for the fixed strategy \eqref{ud3k6AlV7q}, the pricing rule \eqref{CrLPP8jKQG}-\eqref{2vGs8ogSts} satisfies the requirements of Definition~\ref{def:pricingrule} and is rational in the sense of Definition~\ref{def:rational}.

It remains to verify optimality of the proposed strategy. The strategy \eqref{ud3k6AlV7q} is inconspicuous by the preceding proposition. Moreover, Proposition~\ref{v0tMLvYb0E} gives
\[
\bP_t^*=\E{\vt\middle|\sF_t^M},
\qquad
{\rm Var}\!\del{\vt\middle|\sF_t^M}=\bfSigma_t^*.
\]
Since \(\bfSigma_T^*={\bf 0}\), it follows that \(\bP_T^*=\vt\) a.s. By \eqref{SWjcJYb4sK}, this is equivalently \(\bfxi_T^*=\vt-\bp_0\) a.s. Therefore Theorem~\ref{thm:insider-value} applies and shows that \(\bX^*\) attains the upper bound \eqref{BPMSIqokAh} among all \(\bfP^*\)-admissible strategies. Hence \(\bX^*\) is optimal in the sense of Definition~\ref{def:optimal}. Together with rationality, this proves that \((\bfP^*,\bX^*)\) is an equilibrium in the sense of Definition~\ref{def:equilibrium}. The remaining assertions of the theorem follow from Propositions~\ref{v0tMLvYb0E}, \ref{cnTt7K54xR}, and the preceding proposition.
\end{proof}

\section{Verification details}
\label{app:verification-proofs}

\subsection{Back--Cocquemas--Ekren--Lioui constant-volatility computation}
\begin{proof}
Let $\bfSigma^*$ be as in \eqref{ui5XZD4G8F}. Fix a perturbation direction $\bfeta \in C^1\del{[0,T], \cS^n_{++}}$ so that the perturbed optimum $\bfSigma^* + \varepsilon \bfeta$ is admissible for all $\varepsilon$ in some neighborhood of $0$. In particular, this means that
\begin{equation}
\label{U5lWDbQAc8}
\int_0^T \dot{\bfeta}(t) \dif t = {\bf 0}.
\end{equation}
Define
$$
J(\bfSigma) = \int_0^T \tr{\sqrt{-\bfsigma \dot\bfSigma(t) \bfsigma}} \dif t.
$$
We first show that $\bfSigma^*$ satisfies the first order condition
\begin{equation}
\label{JxQ3NZnX26}
0 = \lim_{\varepsilon\downarrow 0}\frac{J(\bfSigma^* + \varepsilon \bfeta) - J(\bfSigma^*)}{\varepsilon}.
\end{equation}
Due to \cite[theorem 1.1]{moralniclas2018}, we have the Taylor expansion
\begin{equation}
\begin{split}
\sqrt{-\bfsigma \del{\dot\bfSigma^*(t) + \varepsilon \dot\bfeta(t)}\bfsigma}
=& \sqrt{-\bfsigma \dot\bfSigma^*(t)\bfsigma} \\
&- \varepsilon \int_0^{\infty}
e^{-k\sqrt{-\bfsigma \dot\bfSigma^*(t)\bfsigma}}
\bfsigma\dot\bfeta(t)\bfsigma
e^{-k\sqrt{-\bfsigma \dot\bfSigma^*(t)\bfsigma}}
\dif k
+ O(\varepsilon^2).
\end{split}
\end{equation}
We can write
\[
\lim_{\varepsilon\downarrow 0}\frac{J(\bfSigma^* + \varepsilon \bfeta) - J(\bfSigma^*)}{\varepsilon}
=
-\int_0^T
\tr{\int_0^{\infty}
e^{-k\sqrt{-\bfsigma \dot\bfSigma^*(t)\bfsigma}}
\bfsigma\dot\bfeta(t)\bfsigma
e^{-k\sqrt{-\bfsigma \dot\bfSigma^*(t)\bfsigma}}
\dif k}
\dif t.
\]
Permuting cyclically within the trace, the RHS integral becomes
\begin{align*}
\int_0^T\tr{\del{\int_0^{\infty}\bfsigma e^{-2k\sqrt{-\bfsigma \dot\bfSigma^*(t)\bfsigma}}\bfsigma \dif k} \dot\bfeta}\dif t
 &= \int_0^T\tr{\frac{1}{2}\del{\bfsigma\del{-\bfsigma\dot\bfSigma^*_t\bfsigma}^{-1/2}\bfsigma} \dot\bfeta}\dif t
\end{align*}
Substituting in \eqref{ui5XZD4G8F}, we get
\begin{equation*}
 \int_0^T\tr{\frac{1}{2T}\del{\bfsigma\del{\bfsigma\bC\bfsigma}^{-1/2}\bfsigma} \dot\bfeta}\dif t = \tr{\frac{1}{2T}\del{\bfsigma\del{\bfsigma\bC\bfsigma}^{-1/2}\bfsigma}\int_0^T\dot\bfeta\dif t} = 0
\end{equation*}
where the last equality is due to \eqref{U5lWDbQAc8}. This proves \eqref{JxQ3NZnX26}. To conclude the proof, observe that $\bfSigma \mapsto J(\bfSigma)$ is concave.
\end{proof}

\subsection{Proof of the common-eigenbasis proposition}
\begin{proof}
We first restrict our search to those $(\bM_t)_{0\le t\le T}$ which are diagonalizable along $\bV$. If $(\bM_t)_{0\le t\le T}$ is one such choice, then
\begin{align*}
\frac{1}{2}\tr{\bC \bM_0} + \frac{1}{2}\E{\int_0^T \tr{\bfsigma^2_t \bM_t^{-1}}\dif t} &= \sum_{i=1}^{n}\del{\frac{1}{2}\Sigma_0^i M^i_0 + \frac{1}{2}\E{\int_0^T \del{\sigma^i_t}^2 \del{M^i_t}^{-1}\dif t}} \\
&\geq \sum_{i=1}^{n}\del{\frac{1}{2}\Sigma_0^i \bar{M}^i_0 + \frac{1}{2}\E{\int_0^T \del{\sigma^i_t}^2 \del{\bar{M}^i_t}^{-1}\dif t}}.
\end{align*}
This is a lower bound for the dual objective, and this lower bound is attained when $\bM_t$ coincides with \eqref{9mRvQhaTPa}. Thus \eqref{9mRvQhaTPa} is the minimizer among all $(\bM_t)_{0\le t\le T}$ within this class.

Next, we show that this class contains the optimum among all admissible $(\bM_t)_{0\le t\le T}$. We turn to the primal formulation \eqref{OxpDCjjac4}. Fix $t \in [0,T]$ and $-\dot{\bfSigma}_t \in \cS^n_{++}$ and assume that $-\dot{\bfSigma}_t$ has sorted eigenvalues $-\dot{\Sigma}^1_t \geq \cdots \geq -\dot{\Sigma}^n_t \geq 0$. If the eigenvalues of $\bfsigma_t$ are also sorted in the order $\sigma^1_t \geq \cdots \geq \sigma^n_t \geq 0$, then we have the upper bound
$$
\tr{\sqrt{-\bfsigma_t \dot{\bfSigma}_t \bfsigma_t}} \leq \sum_{i=1}^{n}\sqrt{-\del{\sigma^i_t}^2\dot{\Sigma}^i_t},
$$
and this upper bound is attained when $\dot{\bfSigma}_t$ has the same eigenbasis as $\bfsigma_t$. This reasoning is true for all $t \in [0,T]$, and we see that we can maximize the integrand \eqref{OxpDCjjac4} $t$-pointwise if we choose $\dot{\bfSigma}_t$ sharing a common eigenbasis with $\bfsigma_t$ at each $t$. But by \eqref{PTYXvuAbrQ}, this eigenbasis is constant in time. Thus the optimal $\dot{\bfSigma}_t$ always has eigenvectors $\bV$. By \eqref{BUSu1MNQ9m}, this means that the optimal $\bM_t$ does too.
\end{proof}

\phantomsection
\small
\addcontentsline{toc}{section}{References}
\bibliographystyle{alpha}
\bibliography{ref}
\end{document}